\documentclass[a4paper, USenglish, 11pt, cleveref, autoref, thm-restate]{article}

\title{Computing Tarski Fixed Points in Financial Networks}

\author{Leander Besting\thanks{Faculty of Computer Science, RWTH Aachen University, Germany. \texttt{leander.besting@rwth-aachen.de}}%
\and Martin Hoefer\thanks{Faculty of Computer Science, RWTH Aachen University, Germany. \texttt{mhoefer@cs.rwth-aachen.de}}%
\and Lars Huth\thanks{Faculty of Computer Science, RWTH Aachen University, Germany. \texttt{huth@algo.rwth-aachen.de}}}

\date{}

\usepackage[dvipsnames]{xcolor}
\usepackage[margin=0.9in]{geometry}
\usepackage{verbatim}
\usepackage{amssymb}
\usepackage{amsmath}
\usepackage{graphicx}
\graphicspath{{./images/}}
\usepackage{hyperref}
\usepackage{amsthm,mathabx}
\usepackage{subcaption}
\usepackage{cite}

\usepackage{pgfplots}
\pgfplotsset{compat=1.18}
\usepgfplotslibrary{fillbetween}

\usepackage{tikz}
\usetikzlibrary{shapes, arrows, arrows.meta, positioning, patterns, decorations.pathreplacing, decorations.pathmorphing, calc}

\usepackage[ruled,vlined,linesnumbered]{algorithm2e}

\usepackage{todonotes}

\newcommand{\R}{\mathbb{R}}

\newcommand{\veca}{\mathbf{a}}
\newcommand{\vecb}{\mathbf{b}}

\newcommand{\vecp}{\mathbf{p}}

\newcommand{\vecr}{\mathbf{r}}

\newcommand{\calF}{\mathcal{F}}

\newcommand{\calA}{\mathcal{A}}

\newcommand{\classNP}{\textsf{NP}}

\newcommand{\classPPAD}{\textsf{PPAD}}
\newcommand{\classFIXP}{\textsf{FIXP}}

\renewcommand{\epsilon}{\varepsilon}
\renewcommand{\leq}{\leqslant}
\renewcommand{\geq}{\geqslant}

\newtheorem{definition}{Definition}
\newtheorem{theorem}{Theorem}
\newtheorem{lemma}{Lemma}
\newtheorem{corollary}{Corollary}
\newtheorem{example}{Example}

\begin{document}

%
%

\maketitle

\begin{abstract}
Modern financial networks are highly connected and result in complex interdependencies of the involved institutions. In the prominent Eisenberg-Noe model~\cite{EisenbergN01}, a fundamental aspect is \emph{clearing} -- to determine the amount of assets available to each financial institution in the presence of potential defaults and bankruptcy. A clearing state represents a fixed point that satisfies a set of natural axioms. Existence can be established (even in broad generalizations of the model) using Tarski's theorem. 

While a \emph{maximal} fixed point can be computed in polynomial time, the complexity of computing other fixed points is open. In this paper, we provide an efficient algorithm to compute a \emph{minimal} fixed point that runs in strongly polynomial time. It applies in a broad generalization of the Eisenberg-Noe model with any monotone, piecewise-linear payment functions and default costs. Moreover, in this scenario we provide a polynomial-time algorithm to compute a \emph{maximal} fixed point. For networks without default costs, we can efficiently decide the existence of fixed points in a given range.

We also study claims trading, a local network adjustment to improve clearing, when networks are evaluated with minimal clearing. We provide an efficient algorithm to decide existence of Pareto-improving trades and compute optimal ones if they exist.
\end{abstract}

\clearpage

\section{Introduction}

Modern financial systems exhibit highly complex debt relationships between their constituents. An important concern in these networks is systemic risk -- after a shock, financial institutions become pressured to pay back (or \emph{clear}) their debt. This leaves some of them in default. Consequently, creditors receiving little or no payments from their defaulting debtors might in turn be unable to meet their own obligations. As such, default can quickly propagate throughout the whole network. This is a realistic concern with the most well-known occurrence being the financial crisis of 2008 (and many other, less severe episodes since then). 

The canonical framework to understand properties of debt clearing in financial networks is the Eisenberg-Noe model~\cite{EisenbergN01}. It has, for example, been used by the European Central Bank in its STAMP\texteuro \, framework for financial stress-testing \cite{stampe}. There are $n$ financial institutions (termed ``banks'' throughout), which are represented by nodes in an edge-weighted, directed graph. There are $m$ edges, each representing a debt claim, with an edge weight expressing the liability of the claim. Banks have (usually non-negative) external assets, which capture funds available for clearing that are not part of the claim network. The basic solution concept in the Eisenberg-Noe model is a \emph{clearing state}, which yields an assignment of assets of banks and payments on each edge that satisfies a set of natural axioms. When a creditor bank is in default, its claims will not be valued by the liability but only the amount that the creditor will pay off in accordance with their legal requirements\footnote{In the United States, the legal framework for this is given by Chapter 11 bankruptcy \cite{chapter11}.}. As one of the axioms, the Eisenberg-Noe model assumes \emph{proportional} debt clearing, i.e., for a bank in default there is a \emph{recovery rate} given by the ratio of total assets available to the total liabilities of the bank. The recovery rate translates directly into proportional payments for all claims where the bank is the creditor. For a more formal discussion of the model, see Section~\ref{sec:model} below.

Clearing states in the Eisenberg-Noe model are Tarski fixed points. The Knaster-Tarski theorem states that for any monotone function mapping a complete lattice to itself, the set of fixed points of the function constitutes a complete lattice. This implies, in particular, that at least one such fixed point exists and that fixed points might not be unique. The complexity of computing Tarski fixed points has been of significant interest recently, especially on the $k$-dimensional grid (see, e.g., \cite{ChenLY23, BranzeiPR25} and the references therein). 

While it is known how to compute the \emph{maximal} fixed point in the Eisenberg-Noe model in polynomial time~\cite{EisenbergN01,RogersV13}, we concentrate on finding every other fixed point (most notably, the \emph{minimal} one). A maximal clearing state is desirable from a centralized perspective, since it yields pointwise maximal assets and payments for each bank and claim, respectively. However, it is unlikely that in the present multi-national financial system, a central coordination agency can dictate such assets and payments in times of crises. As such, arguably, the emergence of other clearing states is more realistic. The \emph{minimal} one has been shown to emerge as the limit of a natural sequential clearing process~\cite{CsokaH18}. More generally, other clearing states can also emerge if agents update their recovery rates in a sequential fashion~\cite{PappW21sequential}. However, to our knowledge, finding efficient algorithms to compute such clearing states or characterizing their computational complexity are important open problems. 

\paragraph{Our Results and Techniques}
We show that a minimal clearing state in the Eisenberg-Noe model can be computed in polynomial time. In contrast to computing the maximal one, it is not possible to formulate the problem directly as an LP. Our approach in Section~\ref{sec:computeMin} can be seen as a careful adjustment of the bottom iteration for Tarski fixed points. A naive implementation of the bottom iteration (which results, e.g., from distributed clearing processes~\cite{CsokaH18}) would start with the external assets at each bank and then propagate payments and assets into the system. This can take infinite time to converge to the fixed point. Instead of propagating all external assets directly, we inject external assets sequentially and simulate their (infinite) propagation by solving several carefully designed LPs. In particular, upon injection of additional assets at a node $v$, our algorithm distinguishes two cases: (1) all additional assets circulate and gradually reach a ``sink'' bank without any solvency event, and (2) some assets keep circulating until another bank becomes solvent. Case (1) yields a linear increase in the payments of every claim, and we compute the slopes by solving a polynomial-sized system of linear equations. Then we can inject the assets at $v$ and compute the increase in assets and payments using the slopes. Case (2) arises if any positive portion of the assets injected at $v$ eventually reaches a part $C$ of the network that we term \emph{flooded} -- where all outgoing paths of any node eventually return to that node. More formally, consider the strongly connected components (SCCs) of the remaining network of open claims. The graph of SCCs is a DAG, and a flooded part $C$ is exactly a sink component of the SCC graph, which contains more than one bank. In this case, we do not inject any additional assets at $v$ but instead increase the payments within the sink component $C$ by a circulation (consistent with proportional payments) until a bank becomes solvent. This can be computed using a suitable LP. We then update the SCC graph and again attempt to inject additional assets at $v$. Overall, we see that in each iteration, we either fully inject the external assets of a bank or create a solvent bank. Thus, at most $2n$ iterations are needed.

Our algorithmic ideas turn out to be very powerful to address many generalizations of the problem. First, we relax the condition of proportional payments to the class of arbitrary \emph{monotone and piecewise-linear} payment functions. This class includes many important examples considered in the literature recently, including edge ranking~\cite{BertschingerHS25,HoeferW22,FroeseHW25} (or singleton liability priorities~\cite{IoannidisKV23}), unit-ranking~\cite{CsokaH18,BertschingerHS25,HoeferW22}, or priority-proportional~\cite{KanellopoulosKZ24, Elsinger09}. By using appropriate granularity, we can handle arbitrary monotone payment functions via approximation with piecewise-linear ones. The extension requires applying our algorithm and the analysis in phases, where each phase ends with reaching an interval border for the payment function of any edge. In each phase, all payments behave linearly, so slopes in case (1) and circulation in case (2) can be obtained by setting up and solving appropriate systems of linear equations. The interval borders of the payment functions limit the increase in external assets of $v$ in case (1), as well as the amount of the circulation in case (2).

Second, we show how to handle an extension of the model to \emph{default costs} for insolvent banks~\cite{RogersV13}. Here, the assets of each defaulting bank are additionally reduced by a constant factor. This adjustment represents a linear decrease in the assets, which can be integrated into the network. Since we monotonically increase payments and assets throughout, we (only) need to be careful when a bank $v$ becomes solvent. The challenge is that the fixed point function ceases to be continuous from below. Thus, a standard bottom iteration might fail to guarantee convergence to a minimal clearing state in the limit. 

For each bank $v$, we can represent default cost using an auxiliary claim to an auxiliary individual ``sink bank'' such that this bank receives the default cost. When $v$ becomes solvent, the default cost vanishes, we remove the auxiliary sink bank, and inject the remaining assets as additional external assets at the out-neighbors of $v$. This view directly indicates how our approach of sequential injection of external assets can be extended to handle default cost. 

Third, for networks without default cost, we can also compute \emph{arbitrary} clearing states. When we consider any clearing state, the difference flow between the clearing state and a minimal clearing state is a circulation. As such, by applying case (2) of the algorithm above and iteratively ``flooding'' sink components (partially, in any order), we can produce any possible clearing state. We use this insight to tackle a \emph{range clearing problem}: For a subset of relevant banks $S \subseteq V$, there is a closed target interval $I_v \subset \R$, for each $v \in S$. The question is whether there exists a clearing state $\veca$ such that $a_v \in I_v$ for all $v \in S$. We show that this problem can be solved in polynomial time.

More fundamentally, in Section~\ref{sec:computeMax} we show an interesting structural equivalence. Every financial network with piecewise-linear payment functions can be transformed into an equivalent network with priority-proportional payments. The transformation increases the representation by at most a polynomial factor. This shows the generality of priority-proportional payment functions. Moreover, it allows applying an existing algorithm for the computation of \emph{maximal} clearing states~\cite{KanellopoulosKZ24} to networks with arbitrary monotone, piecewise-linear payments and default cost. However, the algorithm in~\cite{KanellopoulosKZ24} is technically not polynomial-time since it relies on pushing adjustment steps to the limit to compute the payments for each iteration. Instead, using our approach we can indeed implement the procedure in polynomial time.

Finally, in Section~\ref{sec:trade} we study \emph{claims trades} as a network adjustment to influence the minimal clearing state. The operation was recently formalized in~\cite{HoeferVW24}. In a claims trade, a given bank $w$ strives to buy a given claim $e=(u,v)$ from a creditor bank $v$. Formally, $w$ should pay a return $\rho$ from its external assets to $v$ and become creditor of claim $e$. The goal is to give liquidity to $v$ and raise the assets of $v$ in order to mitigate contagion effects as much as possible. In contrast to previous work, we focus on the scenario when the network is evaluated with a minimal clearing state. 

Ideally, one would want to compute a return $\rho$ such that \emph{both} $v$ and $w$ \emph{strictly} benefit from the trade (w.r.t.\ the minimal clearing state). We show that this is impossible for all networks with strictly monotone payment functions. We show a novel characterization for banks, for which the clearing state is not unique (Lemma~\ref{lem:characterize}). We then focus on creditor-positive trades, in which $v$ strictly profits but $w$ stays at least indifferent. For networks with piecewise-linear payment functions, we show how to decide in polynomial time if a creditor-positive claims trade exists, and how to compute one that maximizes the post-trade assets of $v$ if it exists. Our proof shows that the set of returns that yield creditor-positive trades forms a consecutive interval. Indeed, the largest such return maximizes the assets of $v$, and we can find it by a combination of binary search and maximization of suitable LPs.


\paragraph{Further Related Work}
Algorithmic aspects of the Eisenberg-Noe model for financial networks have become a popular topic in recent years. A \emph{maximal} clearing state can be computed in polynomial time for proportional payments~\cite{EisenbergN01}, even more generally in networks with default cost~\cite{RogersV13} and with priority-proportional payments~\cite{KanellopoulosKZ24}. It is known that decentralized update procedures (which essentially implement a bottom-iteration) converge to the \emph{minimal} clearing state~\cite{CsokaH18}, but these approaches do not necessarily run in polynomial time. Up to our knowledge, an efficient algorithm for minimal clearing has been derived only for networks with edge-ranking payments without default cost~\cite{BertschingerHS25}.

The complexity of computing clearing states has also been considered in an extension of the model to credit-default swaps~\cite{SchuldenzuckerSB17,SchuldenzuckerSB20}, where clearing states exist but their entries are not necessarily rational. Some notions of approximation yield \classPPAD-hardness results~\cite{SchuldenzuckerSB17}, even for constant approximation factors~\cite{DohnHK25,IoannidisKV23PPAD}. Stronger notions of approximation even give rise to \classFIXP-hardness~\cite{IoannidisKV22,IoannidisKV23}. 

Claims trades have recently been studied in the Eisenberg-Noe model without default cost and with \emph{maximal} clearing~\cite{HoeferVW24}. Since a claims trade can be interpreted as a debt swap operation~\cite{FroeseHW25,PappW21EC}, it is impossible that both creditor and buyer banks profit strictly. An optimal creditor-positive trade can be computed in polynomial time for proportional and edge-ranking payment functions. For general functions, it is shown how to compute approximately optimal trades. When trading either multiple incoming or multiple outgoing claims of a single bank, finding optimal trades becomes \classNP-hard. For incoming claims, there is a bicriteria approximation for all monotone payment functions, for outgoing the problem is \classNP-hard to approximate within polynomial factors even for edge-ranking functions. These results were extended and improved in a model with \emph{fractional} claims trades, for networks with proportional payments and default cost~\cite{HoeferHW25}.

Our paper is related to a growing body of work that studies structural and algorithmic properties of game-theoretic scenarios based on the Eisenberg-Noe model. Aspects that have received attention include, for example, strategic payment allocation~\cite{BertschingerHS25,KanellopoulosKZ24,HoeferW22,PappW20}, forgiving, canceling or forwarding debt~\cite{KanellopoulosKZ25,KanellopoulosKZ23,TongKV24Cancel}, donations~\cite{TongKV24}, prepayments~\cite{ZhouWVBSCW24}, lending~\cite{EgressyPW24}, and more~\cite{BertschingerHKL23}. With the exception of~\cite{HoeferW22}, all these works consider only maximal clearing states.

More generally, there is work on the complexity of improvement measures for the clearing properties in Eisenberg-Noe financial networks, such as debt swapping~\cite{PappW21EC,FroeseHW25} or portfolio compression. In portfolio compression, a debt cycle is eliminated from the network. For proportional payments, this can have counter-intuitive effects~\cite{SchuldenzuckerS20,Veraart22}, and many optimization questions surrounding this operation are \classNP-hard~\cite{AminiF23}.

\section{Model and Preliminaries}
\label{sec:model}
\paragraph{Network Model}
We define a financial network \(\calF = (V,E,\mathbf{\ell}, \mathbf{a^x})\). There is a set $V$ of $n$ banks, and a set $E$ of $m$ directed edges. Each edge $e = (u,v) \in E$ has a non-negative edge weight $\ell_e$ that represents the \emph{liability} of a claim with debtor $u$ and creditor $v$. Each bank $v$ has non-negative \emph{external assets} $a^{(x)}_v \ge 0$. For simplicity, we assume that the network has no self-loops or multi-edges\footnote{These aspects can be incorporated by increasing the notational overhead. We leave the straightforward extension to the interested reader.}. We define the set of \textit{outgoing} and \textit{incoming} claims of $v$ by $E^+(v) = \{ e=(v,u) \in E \mid u\in V\}$ and $E^-(v) := \{ e=(u,v) \in E \mid u\in V\}$, respectively. The \emph{total outgoing} and \textit{incoming liabilities} of bank $v$ are  $L^+(v) = \sum_{e \in E^+(v)} \ell_e$ and $L^-(v) = \sum_{e \in E^-(v)} \ell_e$, respectively.

\paragraph{Payment Functions}
The basic solution concept in the Eisenberg-Noe model is a \emph{clearing state}, which defines a consistent set of assets $a_u \ge 0$ for each bank $u \in V$. In the standard model, each bank is assumed to clear debt \emph{proportionally}. Thus, with assets $a_u$, we have a \emph{recovery rate} $\min(1, a_v / L^+(v))$, such that each claim $e =(u,v)$ is cleared to the same fractional extent, i.e., $p_e(a_u) = \min(1, a_v / L^+(v)) \cdot \ell_e$. Using these payments, the assets have to satisfy the natural asset axioms: The \emph{(total) assets} of each bank are given by the external assets plus the incoming payments of other banks, $a_v^{(x)} + \sum_{(u,v) \in E^-(v)} p_{(u,v)}(a_u)$ for each $v \in V$. In particular, this implies that no bank will conduct fraud by generating money or holding back assets from paying its open claims.

We consider several extensions of the model. First, we consider the model with \emph{default cost}~\cite{RogersV13}. Each bank $v$ has a \emph{default rates} $\alpha_v, \beta_v \in [0,1]$. If the bank is insolvent, the available assets $a_v$ that can be used to clear debt are reduced to $a_v^{(\alpha,\beta)}(\veca) = \alpha_v a_v^{(x)} + \beta_v \sum_{(u,v) \in E^-(v)} p_{(u,v)}(a_u)$. The asset axioms for the vector of total assets $\veca = (a_v)_{v \in V}$ of the banks become
\begin{equation}
    \label{eqn:assetAx}
    a_v = \begin{cases} a_v^-(\veca) & \text{if $a_v^-(\veca) \ge L^+(v)$ (i.e., $v$ solvent), and}\\
                        a_v^{(\alpha,\beta)}(\veca) & \text{otherwise.}\end{cases}
\end{equation}
with the \emph{incoming assets} of $v$ (before default cost reduction) given by
\[
    a_v^{-}(\veca) = a^{(x)}_v + \sum_{(u,v) \in E^-(v)} p_{(u,v)}(a_u). 
\]
and the reduced assets due to default cost given by
\[
    a_v^{(\alpha,\beta)}(\veca) = \alpha_v a^{(x)}_v + \beta_v \sum_{(u,v) \in E^-(v)} p_{(u,v)}(a_u). 
\]
We recover the standard model without default cost when $\alpha_v = \beta_v = 1$ for all $v \in V$.

\newcommand{\RR}{\mathbb{R}}
Second, we address a general class of monotone payments~\cite{BertschingerHS25}. Each bank $u$ has a \emph{payment function} $p_e : \RR \to \RR$ for each claim $e = (u,v)$. A payment function satisfies, for each $e \in E$, $u \in U$ and $a_u,\varepsilon > 0$
\begin{equation}
\label{eqn:payAx}
    \begin{aligned}
        p_e(a_u) &\in [0,\ell_e]  && \text{(no edge under- or overpaid)}\\
        p_e(a_u) &\le p_e(a_u+\epsilon) && \text{(monotonicity)}\\
        \sum_{e \in E^+(u)} p_e(a_u) &= \min\{a_u, L^+(u)\}  && \text{(no fraud)}
    \end{aligned}
\end{equation}
These constraints imply that all $p_e$ must be continuous (since $p_e(a + \delta) - p_e(a) \leq \delta$ for any $\delta > 0$). The axioms are trivially fulfilled for proportional payments. Many other natural examples of payment functions from this class have been considered, e.g., priority/edge-ranking~\cite{BertschingerHS25,CsokaH18} (rank edges in an order and pay them sequentially until running out of assets), constrained equal awards or losses~\cite{CsokaH24} (all claims receive the same payment or the same non-payment, up their liability) or priority-proportional~\cite{KanellopoulosKZ24} (partition edges into sets, rank sets in an order, pay edges proportionally within each set, and sequentially over sets until running out of assets). All these examples share a natural property -- they are \emph{piecewise-linear}. In this paper, we concentrate on the class of monotone, piecewise-linear payment functions.

\begin{definition}
    A \emph{piecewise-linear} payment function $p_e : \RR \to \RR$ for an edge $e \in E^+(v)$ is given by $k_e \ge 1$ interval borders $0 = x_{e,0} < x_{e,1} < \ldots < x_{e,k_e} = L^+(v) < x_{e,k_e+1} = \infty$ and slopes $m_{e,i} \ge 0$ for each $i=1, \ldots, k_e+1$ such that
    $$  p_e(a) = m_{e,i} \cdot (a - x_{e,i-1}) + p_e(x_{e,i-1}) \qquad \text{ for any } a \in [x_{e,i-1}, x_{e,i}).$$
    Moreover, $p_e$ adheres to the three axioms in~\eqref{eqn:payAx}. 
\end{definition}
Note that $p_e(a) = \ell_e$ when $a \ge L^+(v) = x_{e,k_e}$. Thus $m_{e,k_e+1} = 0$. We further define $m_e(a)$ as the slope $m_{e,i}$ that applies for argument $a$, i.e., $m_e(a) = m_{e,i}$ such that $a \in [x_{e,i},x_{e,i+1})$. This implies that, for every $v \in V$,
\begin{equation}
    \label{eq:rowStochastic}
    \sum_{e \in E^+(v)} m_{e}(a) = \begin{cases} 1 & \text{ for all } 0 \le a < L^+(v) \\
    0 & \text{ for all } a \ge L^+(v). \end{cases}
\end{equation}
For each bank's available assets, we define the additional amount of assets until a new interval is reached by $\delta_e(a) = \min \{x_{e,i} - a \mid x_{e,i} > a\}$. We denote the total number of interval borders in $\calF$ by $k = \sum_{e \in E} k_{e} \ge m$. Finally, for a given vector of total assets $\veca$, we call an edge $e = (u,v)$ \emph{active} if $m_e(a_u) > 0$. More generally, the \emph{set of active edges} for $\veca$ is $E_\veca = \{ e \in E \mid e$ is active for $\veca \}$, and the \emph{active graph} is $G_\veca = (V,E_\veca)$. For each active edge $e \in E_\veca$, the \emph{active interval} is the index $i$ such that $p_e(a_v) \in [x_{e,i},x_{e,i+1})$.

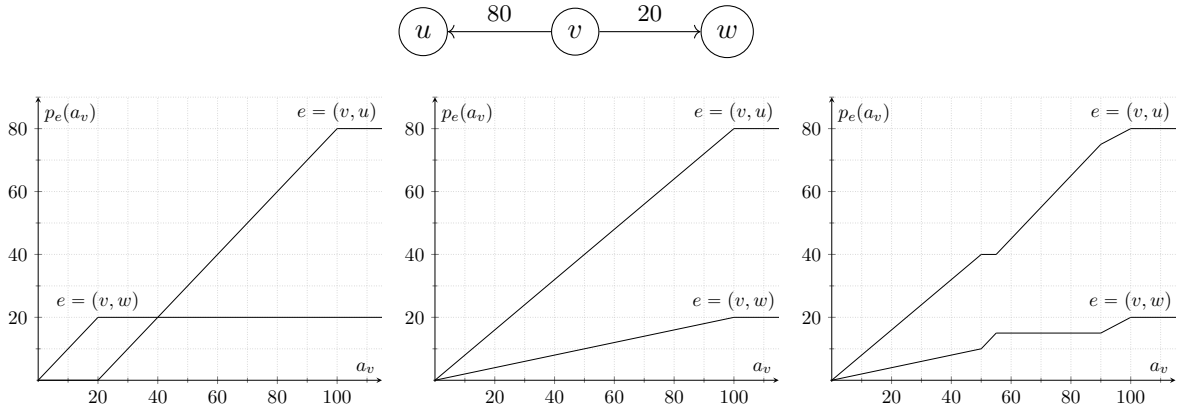
\begin{figure}
    \begin{subfigure}{.31\linewidth}
        \hspace{10pt}
    \end{subfigure}
    \begin{subfigure}{.22\linewidth}
        \begin{tikzpicture}
            \node[draw, circle] (u) at (-2,0) {\(u\)};
            \node[draw, circle] (v) at (0,0) {\(v\)};
            \node[draw, circle] (w) at (2,0) {\(w\)};
        
            \path[->]
                (v) edge node[pos=0.5,above] {\footnotesize\(80\)} (u)
                (v) edge node[pos=0.5,above] {\footnotesize\(20\)} (w);
        \end{tikzpicture}
    \end{subfigure}
    
    \begin{subfigure}{.31\linewidth}
        \vspace{10pt}
    \end{subfigure}
    
    \begin{subfigure}{.31\linewidth}
        \resizebox{\linewidth}{!}{%
            \begin{tikzpicture}
                \begin{axis}[
                    axis lines=middle,
                    axis line style={-stealth},
                    xmin=-.5, xmax=115, ymin=-.5, ymax=90,
                    xtick distance=20, 
                    ytick distance=20,
                    minor x tick num=1, 
                    minor y tick num=1,
                    xlabel=$a_v$,
                    ylabel=$p_e(a_v)$,
                    grid=both, 
                    grid style={thin, densely dotted, black!20},
                    minor grid style={thin, densely dotted, black!20} 
                        ]
                    \addplot[domain=0:20, name path=emptyduv] {0};
                    \addplot[domain=20:100, name path=inclineuv] {x - 20} node[above]{$e = (v,u)$};
                    \addplot[domain=0:20, name path=inclineuw] {x} node[above]{$e = (v,w)$};
                    \addplot[domain=100:120, name path=filleduv] {80};
                    \addplot[domain=20:120, name path=filledeuw] {20};
                \end{axis}
            \end{tikzpicture}
        }
    \end{subfigure}
    \begin{subfigure}{.31\linewidth}
        \resizebox{\linewidth}{!}{%
            \begin{tikzpicture}
                \begin{axis}[
                    axis lines=middle,
                    axis line style={-stealth},
                    xmin=-.5, xmax=115, ymin=-.5, ymax=90,
                    xtick distance=20, 
                    ytick distance=20,
                    minor x tick num=1, 
                    minor y tick num=1,
                    xlabel=$a_v$,
                    ylabel=$p_e(a_v)$,
                    grid=both, 
                    grid style={thin, densely dotted, black!20},
                    minor grid style={thin, densely dotted, black!20} 
                        ]
                    \addplot[domain=0:100, name path=inclineuv] {0.8*x} node[above]{$e = (v,u)$};
                    \addplot[domain=0:100, name path=inclineuw] {0.2*x} node[above]{$e = (v,w)$};
                    \addplot[domain=100:120, name path=filleduv] {80};
                    \addplot[domain=100:120, name path=filledeuw] {20};
                \end{axis}
            \end{tikzpicture}
        }
    \end{subfigure}
    \begin{subfigure}{.31\linewidth}
        \resizebox{\linewidth}{!}{%
            \begin{tikzpicture}
                \begin{axis}[
                    axis lines=middle,
                    axis line style={-stealth},
                    xmin=-.5, xmax=115, ymin=-.5, ymax=90,
                    xtick distance=20, 
                    ytick distance=20,
                    minor x tick num=1, 
                    minor y tick num=1,
                    xlabel=$a_v$,
                    ylabel=$p_e(a_v)$,
                    grid=both, 
                    grid style={thin, densely dotted, black!20},
                    minor grid style={thin, densely dotted, black!20} 
                        ]
                    \addplot[domain=0:50, name path=uv1] {0.8*x};
                    \addplot[domain=0:50, name path=uw1] {0.2*x};
                    \addplot[domain=50:55, name path=uv2] {x-40};
                    \addplot[domain=50:55, name path=uw2] {40};
                    \addplot[domain=55:90, name path=uv3] {15};
                    \addplot[domain=55:90, name path=uw3] {x-15};
                    \addplot[domain=90:100, name path=uv4] {(0.5*x) + 30} node[above]{$e = (v,u)$};
                    \addplot[domain=90:100, name path=uw4] {(0.5*x) - 30} node[above]{$e = (v,w)$};
                    \addplot[domain=100:120, name path=filleduv] {80};
                    \addplot[domain=100:120, name path=filledeuw] {20};
                \end{axis}
            \end{tikzpicture}
        }
    \end{subfigure}
    \caption{\label{fig:example} Consider the financial network displayed at the top. Payments of bank $v$ on each edge for edge-ranking payment functions with $(v,w)$ ranked first (left); proportional payment functions (middle); piecewise-linear functions with interval borders $0, 50,55,90,100,\infty$ (right).}
\end{figure}

\paragraph{Clearing States}
Given a financial network $\calF$ and piecewise-linear payment functions $\vecp = (p_e)_{e \in E}$, a \emph{clearing state} is simply a vector of assets $\veca = (a_v)_{v \in V}$ such that all asset axioms~\eqref{eqn:assetAx} are fulfilled. These axioms are fixed point conditions, since the value of $a_v$ depends on other $a_u$ (and potentially vice versa). Consider the space of all potential asset vectors $\mathcal{A} = \bigtimes_{v \in V} [0,L^-(v)+a_v^{(x)}]$. We define a function $\Phi_{\veca^{(x)}}(\veca) : \mathcal{A} \to \mathcal{A}$ resulting from applying the map defined by the asset axioms~\eqref{eqn:assetAx}. $\Phi_{\veca^{(x)}}$ is monotone since all $p_e$ are monotone. More formally, if $\veca' \ge \veca$ coordinate-wise, then $\Phi_{\veca^{(x)}}(\veca') \ge \Phi_{\veca^{(x)}}(\veca)$ coordinate-wise. The fixed points of $\Phi$ are the clearing states. Applying the Knaster-Tarski theorem, we see that the set $\mathcal{A}_f$ of fixed points and the coordinate-wise $\ge$-operation form a complete lattice.

We denote by $\hat{\veca}$ the minimal fixed point and by $\check{\veca}$ the maximal one. A natural attempt towards computation of $\hat{\veca}$ is a \emph{bottom iteration}: Starting from any $\veca^0 \le \hat{\veca}$ (say, $a_v^0 = a_v^{(x)}$ for all $v \in V$) we iteratively compute $\veca^{i+1} = \Phi_{\veca^{(x)}}(\veca^{i}) \le \hat{\veca}$. Monotonicity directly implies that $\veca^{i+1} \ge \veca^{i}$. In the model without default cost (all default rates $\alpha_v = \beta_v = 1$) and arbitrary monotone payment functions, it is easy to see that $\lim_{i\to \infty} \veca^i = \hat{\veca}$, but there are simple instances where $\veca^i \neq \hat{\veca}$ for every $i \ge 0$. More generally, when $\alpha_v$ or $\beta_v < 1$ for some $v \in V$, there are simple instances where $\lim_{i\to \infty} \veca^i \neq \hat{\veca}$. For example, bank $v$ can be solvent in $\hat{\veca}$ but insolvent in $\veca^i$, for every $i \ge 0$. Then the limit of the total assets of $v$ is $\lim_{i \to \infty} a_v^i < L^+(v) \le \hat{a}_v$. For small examples of the bottom iteration, see Examples~\ref{exmpl:bottom iteration} and~\ref{exmpl:bottom iteration} in Appendix~\ref{app:examples}.

\paragraph{Claims Trades}
An operation to improve clearing states are claims trades. In a claims trade, we are given a claim $e = (u,v)$ and a potential buyer bank $w$. The bank $w$ buys the claim by paying an amount of external assets $\rho$ to $v$. We call $\rho$ the \emph{return}. In turn, the creditor of the claim $e$ is changed from $v$ to $w$. We assume that the return is upper bounded by $\rho \le \min\{a^{(x)}_w, \ell_{(u,v)}\}$. After the trade, the external assets of $w$ are $a_w^{(x)} - \rho$ and the ones of $v$ are $a_v^{(x)} + \rho$. A claims trade can represent a \emph{donation}, in which $w$ only transfers $\rho$ external assets without changing any edges in the network\footnote{To formulate donations as special cases of claims trades, we may simply assume that $w$ trades an auxiliary claim $(u,v)$ with liability $a^{(x)}_w$ and no existing payments from an auxiliary debtor bank $u$ of $v$.}.

We consider claims trades when the network is evaluated with $\hat{\veca}$. A claims trade is called \emph{creditor-positive} if there exists a return $\rho$ such that the assets of creditor bank $v$ are strictly improved and the assets of buyer bank $w$ remain at least the same. We call $\rho$ a \emph{creditor-positive return}. If it exists, we look for an \emph{optimal} one, i.e., a creditor-positive return $\rho^*$ that maximizes the post-trade assets of $v$. Note that any creditor-positive trade Pareto-improves the assets in the entire network.

\section{Computing Clearing States}
\label{sec:computeMin}


\newcommand{\vecZero}{\mathbf{0}}
\newcommand{\vecs}{\mathbf{s}}
\newcommand{\vecd}{\mathbf{d}}
\newcommand{\vece}{\mathbf{e}}
\newcommand{\matM}{\mathbf{M}}
\newcommand{\matI}{\mathbf{I}}
\newcommand{\calP}{\mathcal{P}}

\begin{algorithm}[t]
\DontPrintSemicolon
\SetKwInOut{Input}{Input}\SetKwInOut{Output}{Output}

\Input{Financial network $\calF$}
\Output{Minimal clearing state $\hat{\veca}$} \medskip
    $\vecb^{(x)} \gets \vecZero$, $\vecb \gets \vecZero$ \tcp*{External and outgoing assets}
    $D \gets \{u \in V \mid \alpha_u < 1 \text{ or } \beta_u < 1 \}$ \label{line:D}\;
    Adjust each bank $u \in D$ with auxiliary banks, reduced external assets, redirected edges \label{line:adjust}\;
    \While{there is bank $v \in V$ with $b^{(x)}_v < a^{(x)}_v$}{
        \smallskip
        \tcp{Repeated flooding of reachable sink-SCCs}
        $\mathcal{C} \gets$ strongly connected components (SCC) of the active graph $G_\vecb$\;
        $G_\mathcal{C} \gets$ directed acyclic graph of SCCs of $G_\vecb$\;
        \While{there is non-singleton $C \in \mathcal{C}$ that is reachable from $v$ and a sink in $G_\mathcal{C}$}{
            Solve Flood-LP~\eqref{eq:FloodLP}, let $\vecd^*$ be the optimal solution\;
            $\vecb \gets \vecb + \vecd^*$\;
            Update $G_\vecb$, $\mathcal{C}$, and $G_\mathcal{C}$\;
        }
        \smallskip
        \tcp{Raise external assets of $v$}
        Solve Increase-LP~\eqref{eq:IncreaseLP}, let $(\delta^*,\vecd^*)$ be an optimal solution\;
        $b^{(x)}_v \gets b^{(x)}_v + \delta^*$ and $\vecb \gets \vecb + \vecd^*$\; 
        \smallskip
        \tcp{Adjust the network for banks with default cost}
        \ForAll{banks $u \in D$ that became solvent\label{line:ifStart}}{
            $b_w^{(x)} \gets b_w^{(x)} + p_e(b_u)$, for each $e = (u,w) \in E$\;
            $a^{(x)}_w \gets a^{(x)}_w + \ell_e$, for each $e = (u,w) \in E$ \;
            Remove all edges $e = (u,w)$ from $\calF$\;
            \label{line:ifEnd}
        }
    }
    \Return{$\vecb$}
    \caption{\label{alg:minClear} Computation of a minimal clearing state}
\end{algorithm}

\subsection{Minimal Clearing State}
In this section, we explain our algorithm to compute a minimal clearing state. The standard approach to computing minimal Tarski fixed points is a bottom iteration, but implementing this directly in financial networks is not effective. Instead of a standard bottom iteration, Algorithm~\ref{alg:minClear} maintains and increases a vector $\vecb^{(x)}$ of \emph{reduced external assets}. It starts at $\vecb^{(x)} = \vecZero$ and approaches the vector of actual external assets from below, i.e., $\vecb^{(x)} \le \veca^{(x)}$. In each iteration of the main while-loop, the algorithm computes an increase of reduced external assets at one vertex $v$. It also maintains a vector $\vecb$ of total assets. For a small example run of the algorithm, see Example~\ref{exmpl:runthrough} in Appendix~\ref{app:examples}.

For the analysis, we will maintain the invariant that at the beginning of an iteration of the main while-loop, $\vecb$ is a minimal clearing state of the network $\calF$ with the reduced external assets $\vecb^{(x)}$.
\begin{lemma}\label{lem:invariant}
    At the beginning of each iteration of the main while-loop, $\vecb^{(x)} \le \veca^{(x)}$ and $\vecb$ is a minimal clearing state of the financial network $\calF$ with external assets $\vecb^{(x)}$.
\end{lemma}
%
For clarity of exposition, we first prove correctness of our algorithm when there is no default cost, i.e., all default rates are $\alpha_v = \beta_v = 1$. We then outline the adjustments for general default rates in Section~\ref{sec:defaultCost} below. In particular, without default cost, $D = \emptyset$ in line~\ref{line:D}, we make no adjustments in lines~\ref{line:adjust}, and we never execute lines \ref{line:ifStart}-\ref{line:ifEnd}. 

For the proof of Lemma~\ref{lem:invariant}, the properties clearly hold in the beginning since $\vecb^{(x)} = \vecb = \vecZero$. Suppose they are true at the beginning of iteration $i$. We argue that the properties hold in the end of iteration $i$, i.e., at the beginning of iteration $i+1$. Our first insight shows that the minimal clearing state is non-decreasing in the external assets of each bank.
\begin{lemma}
    \label{lem:monotone}
    Consider a financial network $\calF$. Suppose all banks have reduced external assets $\vecZero \le \vecb^{(x)} \le \veca^{(x)}$. If $a^{(x)}_u > b^{(x)}_u$, then in the minimal clearing state $\hat{\vecb}$ resulting from $\vecb^{(x)}$ we have $\hat{a}_u > \hat{b}_u$.
\end{lemma}
\begin{proof}
    For contradiction, assume that $\hat{a}_u \le \hat{b}_u$. Consider the pointwise minimum $\vecb' = \hat{\veca} \land \hat{\vecb}$, i.e., $b'_v = \min\{\hat{a}_v, \hat{b}_v\}$ for all $v \in V$. Since $\vecb' \le  \hat{\vecb}$, the bottom iteration shows that $\Phi_{\vecb^{(x)}}(\vecb') \ge \vecb'$. Now $\hat{\veca} \ge \vecb'$, $a^{(x)}_u > b^{(x)}_u$ and $\Phi$ is strictly monotone with respect to external assets and weakly monotone with respect to total assets. This implies that 
    \[ 
        \Phi_{\veca^{(x)}}(\hat{\veca})_u > \Phi_{\vecb^{(x)}}(\hat{\veca})_u \ge \Phi_{\vecb^{(x)}}(\vecb')_u \ge \vecb'_u = \hat{a}_u.
    \]
    Thus, $\hat{\veca}$ is not a clearing state -- a contradiction.
\end{proof}
\paragraph{Flooding SCCs}
Now consider the bank $v$ with $b^{(x)}_v < a^{(x)}_v$ chosen for the increase in iteration $i$ of the main while-loop. Let us concentrate on regions of the active graph $G_\vecb$ that are reachable from $v$ and strongly connected. 
\begin{definition} 
    We define a \emph{phase} as a subset of asset vectors $\calP \subseteq \calA$ such that for all $\veca,\veca' \in \calP$ we have the same active edges $E_\veca = E_{\veca'}$, and for each $e \in E_\veca$, the same active interval in $\veca$ and $\veca'$.
\end{definition}
\begin{definition}
    We say \emph{$v$ causes a flood in $G_\vecb$} if there is a strongly connected component $C \subseteq V$ of $G_\vecb$ that is (1) reachable from $v$, (2) non-singleton, and (3) a sink-component, i.e., does not have outgoing edges to banks outside $C$.
\end{definition}
\begin{lemma}
    \label{lem:flood}
    Suppose we increase the external assets of bank $v$ from $b^{(x)}_v$ to $b^{(x)}_v + \varepsilon$ for some positive amount $0 < \varepsilon \le a^{(x)}_v - b^{(x)}_v$. Let $\vecb'$ be the new minimal clearing state after the increase. If $v$ causes a flood in $G_\vecb$, then $\vecb'$ cannot be in the same phase as $\vecb$, for any $\varepsilon > 0$.
\end{lemma}
\begin{proof}
    By Lemma~\ref{lem:monotone}, we know that in the minimal clearing states $b'_v > b_v$. As a consequence, each active outgoing edge $e = (v,w)$ must have strictly higher payments in $\vecb'$, i.e., $p_e(b'_v) > p_e(b_v)$ because the slope $m_e(b_v) > 0$. This shows
    \begin{align*}
        b'_w &= b^{(x)}_w + p_e(b'_v) + \sum_{(u,w) \in E} p_{(u,w)}(b'_u) \ge b^{(x)}_w + p_e(b'_v) + \sum_{(u,w) \in E} p_{(u,w)}(b_u) \\
        &> b^{(x)}_w + p_e(b_v) + \sum_{(u,w) \in E} p_{(u,w)}(b_u) = b_w
    \end{align*}
    Consider the sets of banks and active edges reachable from $v$ in $G_\vecb$. Applying the insight inductively shows that every reachable bank has strictly higher total assets, and every reachable active edge strictly higher payments in $\vecb'$ than in $\vecb$, respectively.

    Now, suppose for contradiction that $\vecb' > \vecb$ and both minimal clearing states are in the same phase. Consider a non-singleton sink-SCC $C$ reachable from $v$. For every bank $w \in C$, we consider $d_w = b'_w - b_w > 0$. $C$ is not a singleton, so all banks $w$ in $C$ are insolvent. $\vecb'$ and $\vecb$ are in the same phase, so \eqref{eq:rowStochastic} implies that all of $d_w$ gets paid to out-neighbors in $C$. In turn, these payments become additional incoming assets at other nodes of $C$. Summing over all additional incoming assets of $w \in C$ from neighbors of $C$, we see that
    \begin{equation}
        \label{eqn:outIn}
        \sum_{w \in C} d_w \le \sum_{w \in C} \sum_{\substack{v'\in C\\(v',w) \in E_\vecb^-(w)}} m_{(v',w)}(b'_{v'}) \cdot d_{v'}
    \end{equation}
    Now let $w$ be a node of $C$ that is closest to $v$, and let $(u,w)$ be an edge on a shortest $v$-$w$-path in $G_\vecb$. Since $u$ is reachable from $v$, we know that $b'_u > b_u$. If $v \in C$, then we can assume $u = v$, and since $C$ is non-singleton, an out-neighbor $w \in C$ with $(v,w) \in E$ must exist. Let $d_{(u,w)} = p_{(u,w)}(b'_u) - p_{(u,w)}(b_u)$, then $d_{(u,w)} > 0$. Since $\vecb'$ is a clearing state, the asset axiom holds for $w$. As such, the additional assets $d_w$ are lower bounded by the sum of additional incoming assets over $(u,w)$ and from in-neighbors of $C$, i.e.,
    \begin{align*}
        d_w \ge d_{(u,w)} + \sum_{\substack{v'\in C\\(v',w) \in E_\vecb^-(w)}} m_{(v',w)}(b'_{v'}) \cdot d_{v'} 
        \; > \; \sum_{\substack{v'\in C\\(v',w) \in E_\vecb^-(w)}} m_{(v',w)}(b'_{v'}) \cdot d_{v'}         
    \end{align*}
    With \eqref{eqn:outIn} this proves that there must be some bank $w' \in C$ such that 
    \[
        d_{w'} < \sum_{\substack{v'\in C\\(v',w') \in E_\vecb^-(w')}} m_{(v',w')}(b'_{v'}) \cdot d_{v'},
    \]
    i.e., the additional assets of $w'$ are \emph{strictly less} than its additional incoming assets from $C$. Thus, $\vecb'$ is not a clearing state -- a contradiction.
\end{proof}

$\vecb'$ cannot be in the same phase as $\vecb$, since raising the external assets at $v$ would leave only inconsistent assignments of assets for every reachable, non-singleton, sink-SCC $C$. Then again, $C$ is reachable and external assets of $v$ must be increased to obtain $\hat{\veca}$, so the assets in $C$ must also grow. Therefore, we have to raise the clearing state $\vecb$ within $C$ to escape the current phase and eventually enable an increase of external assets at $v$. In Algorithm~\ref{alg:minClear}, we raise the assets in $C$ by a minimal circulation that maintains the properties of a clearing state (i.e., additional incoming assets are additional outgoing assets) and suffices to advance to the next phase. This is achieved by solving the following LP. We term this a \emph{flooding} of $C$ with assets. Recall that $\delta_{e}(a)$ is the smallest amount of additional assets required to advance the payment function $p_{e}$ to the next interval. 
\begin{equation}
    \label{eq:FloodLP}
    \begin{aligned}
        \text{Max.\ } \; &\displaystyle \sum_{w \in C} d_w &\\
        \text{s.t.\ } \; & d_w = \sum_{\substack{v'\in C\\(v',w) \in E_\vecb^-(w)}} m_{(v',w)}(b_{v'}) \cdot d_{v'} & \forall w\in C\\
        & d_w \le \delta_e(b_w) & \forall e\in E_\vecb^+(w)\\
        & d_w \ge 0 & \forall w\in C
    \end{aligned}
\end{equation}

\paragraph{Solving LP~\eqref{eq:FloodLP}} 
Let us define the weighted $|C| \times |C|$ adjacency matrix $\matM$ for component $C$ with entries $m_{v',w} = m_{(v',w)}(b_{v'})$. $C$ being a SCC and property \eqref{eq:rowStochastic} imply that $\matM$ is row-stochastic and irreducible. The first set of constraints can be written as $\vecd = \vecd \matM$. Thus, $\vecd$ is an eigenvector of $\matM$ with eigenvalue 1. The Perron-Frobenius theorem implies $\vecd$ is non-negative and unique. Thus, we can construct $\matM$, compute $\vecd$ and scale it to the largest multiple such that all constraints $d_w \le \delta_e(b_w)$ are satisfied. Note that power iteration methods for approximating $\vecd$ might not be applicable since $\matM$ is not necessarily aperiodic.

\paragraph{Increasing External Assets}
Flooding of components monotonically increases the assets and changes the phase, but leaves the external assets of $v$ untouched. Thus, after a finite number of repetitions, we must reach a clearing state $\vecb < \hat{\veca}$ such that all sink-SCCs reachable from $v$ are singletons, i.e., solvent banks. In this case, the next lemma shows that there exists a sufficiently small $\delta > 0$ such that for any increase of external assets of $v$ to less than $b^{(x)}_v + \delta$, the resulting minimal clearing state $\vecb'$ remains in the phase of $\vecb$.
 
\begin{lemma} 
    \label{lem:increase}
    Suppose we increase the external assets of bank $v$ from $b^{(x)}_v$ to $b^{(x)}_v + \varepsilon$ for some positive amount $0 < \varepsilon \le a^{(x)}_v - b^{(x)}_v$. Let $\vecb'$ be the new minimal clearing state after the increase. If $v$ does not cause a flood in $G_\vecb$, there is a $\delta > 0$ such that $\vecb'$ is in the same phase as $\vecb$, for every $\varepsilon < \delta$.
\end{lemma}

\noindent {\color{darkgray}\textbf{\textsf{Proof.}}}
    Similar to our observation above, we consider the additional assets $d_w = b'_w - b_w \ge 0$. Since the external assets of $v$ are strictly increased, we see that $d_v > 0$. As $\vecb'$ is a minimal clearing state, we can assume all banks $u$ that are not reachable from $v$ maintain their current assets and have $d_u = 0$, as they remain unaffected from the increase in external assets at $v$. Assuming that $\vecb'$ and $\vecb$ are in the same phase, the largest increase in the external assets of $v$ is given by LP~\eqref{eq:IncreaseLP} below. The first set of equalities asserts that the increase in assets of $v$ is given by the additional external assets and the additional incoming assets over edges from $E_\vecb^-(v)$. Similarly, for banks that are reachable from $v$, the increase in assets is given by the increase in incoming assets. Clearly, $a_v^{(x)} - b_v^{(x)}$ is a trivial upper bound on the maximum increase. To stay in the same phase, we ensure that the (open) active interval on each edge remains the same, which yields a (strict) upper bound of $\delta_e(b_w)$ for each $e = (u,w) \in E_\vecb$. Using weak inequalities, the optimum solution represents the supremum $\delta$ as stated in the lemma. It represents the smallest increase to advance the minimal clearing state to the next phase when increasing the external assets of $v$.
    \begin{equation}
        \label{eq:IncreaseLP}
        \begin{aligned}
           \text{Max.\ } \; &\delta &\\
           \text{s.t.\ } \; & d_v = \delta + \sum_{(u,v) \in E^-_\vecb(v)} m_{(u,v)}(b_u) \cdot d_u \\
           & d_w = \sum_{(u,w) \in E^-_\vecb(w)} m_{(u,v)}(b_u) \cdot d_w & \text{for all $w \neq v$ reachable from $v$} \\
           & d_w = 0 & \text{for all $w \neq v$ unreachable from $v$}\\           
           & d_w \le \delta_e(b_w) & \forall e\in E_\vecb^+(w)\\
           & \delta \le a^{(x)}_v - b^{(x)}_v \\
           & d_w \ge 0 & \forall w\in V
        \end{aligned}
    \end{equation}

    To show correctness, we argue that the LP indeed allows a unique optimal solution $(\delta^*,\vecd^*)$. We can restrict attention to the set $U \subseteq V$ of banks that are reachable from $v$ in $G_\vecb$. Then the first two constraints in \eqref{eq:IncreaseLP} above compose a system of linear equations describing the asset increase that can be expressed by
    \begin{equation}
        \label{eq:linearIncrease}
        \vecd = \delta \cdot \vece_v + \vecd \; \matM.
    \end{equation}
    Here $\vecd$ is an $|U|$-dimensional row vector with entries $d_u$, $\vece_v$ is an $|U|$-dimensional unit row vector with entry $1$ for $v$ and 0 otherwise, and $\matM$ is an $|U|$-dimensional square matrix with entries $m_{u,w} = m_{(u,w)}(b_u)$ for all $(u,w) \in E_\vecb, u \in U$, and 0 otherwise. The vector $\vecd$ can be given by
    \[
        \vecd = (\matI - \matM)^{-1} \; \delta \cdot \vece_v\enspace.
    \]
    Let us observe that the inverse exists. By \eqref{eq:rowStochastic}, $\matI - \matM$ is weakly diagonally dominant. Since $v$ does not cause a flood, every bank $u \in U$ has a path to a solvent bank in $G_\vecb$. Rows corresponding to solvent banks are strictly diagonally dominant. These properties give rise to a chained variant of diagonal dominance~\cite{AzimradehF16} and imply that $\matI-\matM$ is invertible~\cite{ShivakumarC74}. 

    Thus, the increase in the assets of the minimal clearing state scales linearly in the increase of $b_v^{(x)}$. Since all active intervals of $\vecb$ are open, for every sufficiently small value of $\delta > 0$, the resulting asset vector $\vecb+\vecd$ indeed remains in the same phase. In particular, let 
    \[
        \vecs = (\matI - \matM)^{-1} \cdot \vece_v 
    \]
    be the vector of slopes. For the supremum for all increases that keep $\vecb+\vecd$ in the same phase, we require $d_u = s_u \delta \le \delta_{(u,w)}(b_u)$ for all $(u,w) \in E_\vecb$, which are the fourth set of constraints in~\eqref{eq:IncreaseLP}. Clearly, we also require $\delta \le a_v^{(x)} - b_v^{(x)}$, the maximal increase in external assets of $v$. The lemma follows using 
    \begin{equation}
        \label{eq:min}
        \delta = \min\{a_v^{(x)} - b_v^{(x)}, \min_{(u,w) \in E_\vecb} \delta_{(u,w)}(b_u)/s_u\} > 0. \qed
    \end{equation}

\paragraph{Solving LP~\eqref{eq:IncreaseLP}}
The proof of the previous lemma shows that to solve the LP, we can compute the set $U$ of reachable banks and set up the matrix $(\matI-\matM)$. We then solve $ (\matI - \matM) \vecs = \vece_v$ (e.g., by Gaussian elimination) to obtain the slopes $\vecs$, and determine $\delta$ by computing the minima in~\eqref{eq:min}.

\paragraph{Correctness and Polynomial Time}
We are now ready to prove Lemma~\ref{lem:invariant} as the main invariant of the algorithm.

\begin{proof}[Proof of Lemma~\ref{lem:invariant}]
    Consider round $i$ of the while-loop and the vertex $v$ chosen for the increase in external assets $b_v^{(x)}$. Suppose $v$ causes a flood. Lemma~\ref{lem:flood} shows that by repeatedly flooding the corresponding sink-SCCs, we maintain a clearing state that remains below any minimal clearing state resulting from any strict increase in $b_v^{(x)}$.

    There can only be a finite number of flooding operations. Afterwards, the proof of Lemma~\ref{lem:increase} reveals that the increase in the minimal clearing state is linear in the increase in external assets of $v$. We execute the smallest increase $\delta$ that either yields $b_v^{(x)} + \delta = a_v^{(x)}$ or changes the phase. In any case, in the beginning of iteration $i+1$, the vector $\vecb$ is again the minimal clearing state for $\calF$ with the larger vector of external assets $\vecb^{(x)}$.
\end{proof}

\begin{theorem}
    \label{thm:alg}
    Algorithm~\ref{alg:minClear} computes the minimal clearing state of a given financial network without default cost in time $O( (n+k) \cdot (n^3 + m))$.
\end{theorem}

\begin{proof}
    Lemma~\ref{lem:invariant} shows that Algorithm~\ref{alg:minClear} computes the correct minimal clearing state. Regarding the running time, computing the active graph can be done in time $O(n + m)$. Finding a reachable, non-singleton, sink-SCC (or verifying that none exists) can be done in time $O(n + m)$ using standard depth-first-search methods. Solving LPs~\eqref{eq:FloodLP} and~\eqref{eq:IncreaseLP} requires polynomial time. The running times for this are dominated by the computation of the eigenvector $\vecd^*$ and solving linear equations for vector $\vecs$, respectively. Each of these requires at most time $O(n^3)$. Upon solution of any of these two LPs, we advance to a new phase or meet the desired external assets of a bank. Thus, the number of times we need to solve an LP is upper bounded by $O(n + k)$. 
\end{proof} 

\subsection{Default Cost}
\label{sec:defaultCost}

We now turn to the extension of Algorithm~\ref{alg:minClear} to banks with default cost. We explain and justify the adjustments made in lines~\ref{line:D}-\ref{line:adjust} and lines~\ref{line:ifStart}-\ref{line:ifEnd} of the Algorithm~\ref{alg:minClear}.

\begin{lemma}
    \label{lem:adjust}
    For each financial network $\calF$ with minimal clearing state $\hat{\veca}$ and default cost, there is a network $\calF'$ without default cost with a minimal clearing state $\hat{\veca}'$ equivalent to $\hat{\veca}$.
\end{lemma}
\begin{proof}
    We implement default cost reductions of all banks in $D$ by adjusting the network $\calF$.

    A bank with $\alpha_u = \beta_u = 0$ makes no payments. Hence, we can w.l.o.g.\ assume it has no outgoing edges in $\calF$. Omitting the default cost reduction for the assets of $u$ is inconsequential for the remaining clearing state.
    
    Now consider bank $u$ with $\alpha_u \cdot \beta_u > 0$ that is insolvent in $\hat{\veca}$. We adjust the network $\calF$ to $\calF'$ as follows. We assume $u$ has external assets $\alpha_u a_u^{(x)} \le a_u^{(x)}$. We create two auxiliary banks $u_s$ and $u_t$. $u_s$ shall collect all incoming payments for $u$. It then directs a $(1-\beta_u)$-portion to an auxiliary sink bank $u_t$ and the remaining $\beta_v$-portion to $u$. Towards this end, all incoming edges $(v,u)$ are re-routed to $(v,u_s)$. There are new edges $e_1 = (u_s,u_t)$ and $e_2 = (u_s,u)$ with liability $\ell_{e_1} = (1-\beta_u) L^+_u$ and $\ell_{e_2} = \beta_u L^+_u$. $u_s$ has a proportional payment function. All other banks and payments remain as in $\calF$. We assume that $u$, $u_s$ and $u_t$ have no default cost in $\calF'$.
    
    In $\calF'$, a default cost reduction of $\alpha_u$ is directly incorporated into the external assets of $u$. For the reduction of $\beta_u$ in incoming payments, the auxiliary bank $u_s$ splits off the default cost and sends the corresponding assets to an auxiliary sink $u_t$. It is straightforward to verify that $\hat{\veca}'$ with $a_v' = a_v$ for each $v \in V$ and the direct extension to the assets of auxiliary banks $a'_{u_s} = \sum_{(v,u) \in E} p_e(a'_u)$ and $a'_{u_t} = (1-\beta_u) a'_{u_s}$ is a clearing state in $\calF'$. Any smaller clearing state $\vecb'$ in $\calF'$ can be mapped back to a smaller clearing state $\vecb = (b'_v)_{v \in V}$ for all $v \in V$ (excluding the auxiliary banks). This proves the statement for insolvent banks.

    For a bank $u \in D$ that is solvent in $\hat{\veca}$, we have $p_e(\hat{a}_u) = \ell_e$ for all $e=(u,v) \in E$. We can remove all outgoing edges $(u,v) \in E$ and raise the external assets of each out-neighbor $v$ by $\ell_{(u,v)}$. This creates an equivalent network $\calF'$ for which $\hat{\veca}' = \hat{\veca}$ remains the minimal clearing state. This proves the statement for solvent banks.
\end{proof}

\begin{corollary}
    Using the adjustment in Lemma~\ref{lem:adjust}, Algorithm~\ref{alg:minClear} computes the minimal clearing state of a given financial network with default cost in time $O((k+n+m)\cdot (n^3 + m))$.
\end{corollary}

\begin{proof}
    In the algorithm, all banks are insolvent initially. Thus, we adjust the network in lines~\ref{line:D}-\ref{line:adjust} according to Lemma~\ref{lem:adjust} and apply the algorithm in network $\calF'$, where all banks of $D$ have been adjusted. The introduction of sink banks $u_t$ shows that, in particular, no (insolvent) bank $u \in D$ can be part of any non-singleton sink-SCC. Thus, the flooding step is relevant only for components of banks without default cost. 

    The representation becomes problematic when a bank $u \in D$ becomes solvent w.r.t.\ the liabilities in $\calF$. Then all outgoing edges $e \in E^+(u)$ must now become fully paid. By Lemma~\ref{lem:adjust} we can represent the payment of $\ell_{(u,v)}$ of $u$ towards edge $(u,v)$ by raising the external assets of out-neighbor $v$ by $\ell_{(u,v)}$. 
    
    This is implemented in lines~\ref{line:ifStart}-\ref{line:ifEnd} as follows. For each $u \in D$ we test if it becomes solvent w.r.t.\ unadjusted network $\calF$ (i.e., checking if $a_u^{(x)} + \sum_{(v,u)} p_e(b_v) \ge L^+_u$). If so, by Lemma~\ref{lem:invariant} it is solvent in $\hat{\veca}$. Thus, we adjust $\veca^{(x)}$ by removing all outgoing edges of $u$ and adding external assets of $\ell_{(u,v)}$ to $a^{(x)}_v$. We execute a similar adjustment for $\vecb$ and $\vecb^{(x)}$: Add the \emph{current payments} $p_{(u,v)}(b_u)$ to $b^{(x)}_v$ for each $(u,v) \in E$. This gives an equivalent representation of $\hat{\veca}$ and $\vecb$ after removal of the out-edges of $u$. The remaining increase in payments then gets executed via an increase in external assets of the (former) out-neighbors of $u$. This aligns directly with the invariant of Lemma~\ref{lem:invariant} and proves that the algorithm remains correct for networks with default cost.

    The asymptotic upper bound on the running time in Theorem~\ref{thm:alg} suffers by at most $m$ additional iterations due to an increase of some external assets by the liability of an incoming edge.
\end{proof}

\subsection{Arbitrary Clearing States}

In this section, we show how Algorithm~\ref{alg:minClear} can be extended to compute arbitrary clearing states for financial networks $\calF$ without default cost. More formally, starting from the minimal state $\hat{\veca}$, we show that any further application of the flooding step leads to larger clearing states. Conversely, for any given clearing state $\veca$, we observe that there exists a sequence of flooding steps starting from $\hat{\veca}$ that results in $\veca$. As such, the set of clearing states is \emph{exactly} the set of states reachable from $\hat{\veca}$ by flooding steps. We obtain an algorithmic framework that (by suitable choice of flooding steps) is capable of computing every clearing state of $\calF$.

As a potential application of this insight, we outline a \emph{range clearing problem}: There is a given financial network $\calF$ and a subset $S \subseteq V$ of banks. For each $v \in S$, there is a given interval $I_v \subseteq [0,a^{(x)}_v+L^-(v)]$ that specifies a desired range of total assets after clearing. The question is whether or not there exists a \emph{range clearing state} $\veca$, i.e., a clearing state such that $a_v \in I_v$ for all banks $v \in S$. We show that this problem can be solved in polynomial time in networks $\calF$ without default cost.

We first discuss the extension of Algorithm~\ref{alg:minClear} and then explain how it can be used to for the range clearing problem.
\begin{corollary}
    \label{cor:allClear}
    Algorithm~\ref{alg:minClear} can be extended to compute any clearing state of a given financial network $\calF$ without default cost in polynomial time.
\end{corollary}
\begin{proof}
    Suppose we are given a clearing state $\veca$. Consider any non-sink SCC $C$ of the active graph $G_\veca$ and the eigenvector $\vecd$ computed by LP~\eqref{eq:FloodLP}. Adding to $\veca$ any multiple $\gamma \vecd$ that is a feasible solution for LP~\eqref{eq:FloodLP}, we maintain the property that $\veca + \gamma \vecd$ is a clearing state. Thus, by repeatedly finding non-sink SCCs, computing the eigenvector $\vecd$, and adding a feasible multiple to $\veca$, we can compute an increasing sequence of clearing states.

    Conversely, consider any clearing state $\veca \ge \hat{\veca}$ and the difference $\vecd' = \veca - \hat{\veca}$. Clearly, $\vecd'$ represents a circulation flow in $\calF$. Moreover, for any bank with $d'_v > 0$, these additional assets must be paid to outgoing edges $e \in E^+(v)$. All payments must adhere to the payment functions $p_e$. As such, there must be payments $\vecd \le \vecd'$ that represent a circulation flow within some non-sink SCC $C$ of $G_{\hat{\veca}}$. Repeating this argument, we see that $\vecd'$ is decomposable into a sequence of SCCs $C$ and corresponding solutions of LP~\eqref{eq:FloodLP}. By a suitable choice of flooding steps, we can reach $\veca$.
\end{proof}

\begin{theorem}
    The range clearing problem can be solved in polynomial time in financial networks $\calF$ without default cost.
\end{theorem}

\begin{proof}
    To decide whether or not the instance allows a range clearing state, we first compute $\hat{\veca}$ using Algorithm~\ref{alg:minClear}. If $\hat{a}_v \in I_v$ for all $v \in S$, then $\hat{\veca}$ is a range clearing state. If there is a bank $v \in S$ such that $\hat{a}_v > I_v$, the minimal clearing state $\hat{\veca}$ exceeds the target interval for $v$. Hence, there is no range clearing state.
    
    In the remainder we consider the case $\hat{a}_v < I_v$ for at least one bank $v \in S$ (and $\hat{a}_w < I_w$ or $\hat{a}_w \in I_w$ for all other banks $w \in S$). For our argument we maintain the invariant that the clearing state $\veca$ yields $a_v \in I_v$ or $a_v < I_v$ for all $v \in S$. Consider any bank $v$ with $a_v < I_v$ and the SCC $C$ of the active graph $G_\veca$ that contains $v$. If $C$ is a sink, it is impossible to further increase the assets of $v$ in a clearing state. There is no target clearing state. Otherwise, we raise the assets within $C$ continuously by a multiple of eigenvector $\vecd$ until one of two events occurs: (1) $a_v + d_v \in I_v$ or (2) the active graph changes. Due to monotonicity and the structure of the payment function, $\vecd$ represents a necessary increase in the assets to achieve $a_v \in I_v$ (which is sufficient if event (1) occurs). Clearly, if this results in $a_w + d_w > I_w$ for another bank $w \in S \cap C$, then the interval conditions for $v$ and $w$ cannot be satisfied simultaneously and no target clearing state exists. Otherwise, we have computed a larger clearing state that maintains our invariant. The algorithm has to terminate after a finite number of iterations, after which it either reaches a range clearing state or verifies that none exists.

    The number of iterations required is linear in the total number of events, which is upper bounded by the total number of breakpoints and $|S| \le n$. As such, the asymptotic bound from Theorem~\ref{thm:alg} continues to hold. Clearly, each iteration can be performed in polynomial time (choose $v$, compute SCC $C$, compute increase $\vecd$, check conditions for all banks in $C \cap S$).
\end{proof}

\subsection{Characterization and Maximal Clearing States}
\label{sec:computeMax}

In this section, we briefly discuss how to compute the \emph{maximal} clearing state $\check{\veca}$. In particular, by computing $\hat{\veca}$ and then applying flooding steps greedily to the largest extent, the arguments in Corollary~\ref{cor:allClear} show that a maximal clearing state $\check{\veca}$ will be reached. Note that the upper bound on the running time given in Theorem~\ref{alg:minClear} applies to this extended procedure as well. 

Instead, in this section, we discuss a more direct approach. We show an equivalence of financial networks with general piecewise-linear payment functions to networks with \emph{priority-proportional} functions.

\begin{definition}
    For a bank $v \in V$, a collection of \emph{priority-proportional} payment functions $(p_e)_{e \in E^+(v)}$ with $p_e : \RR \to \RR$ is given by a partition of $E^+(v)$ into $k_v \ge 1$ sets $(E_1,\ldots,E_{k_v})$ such that $E_i \cap E_j = \emptyset$ for $i \neq j$, $\bigcup_i E_i = E^+(v)$, and for every $i=1,\ldots,k_v$ and $e \in E_i$,
    $$  p_e(a) = \begin{cases} 0 & \text{if } a < x_{v,i-1}\\
    \frac{\ell_e}{\sum_{e' \in E_i} \ell_{e'}} \cdot (a - x_{v,i-1}) & \text{if } a \in [x_{v,i-1}, x_{v,i}) \\
    \ell_e & \text{if } a \ge x_{v,i},\end{cases} $$
    where $x_{v,i} = \sum_{j \le i} \sum_{e' \in E_j} \ell_{e'}$ and $x_{v,k_v+1} = \infty$.
\end{definition}

Using priority-proportional payment functions, $v$ clusters its outgoing edges into sets of decreasing priority. It starts to pay all edges of class $E_1$ proportionally until they are fully paid. Then it uses any remaining assets to pay class $E_2$ proportionally, and so on. Our main insight is that in terms of clearing, the class of networks with priority-proportional functions is equivalent to the class of all networks with piecewise-linear ones.

\begin{lemma}
    \label{lem:piecewiseLinearIsPriorityProportional}
    For any financial network $\calF$ with piecewise-linear payment functions, there is a network $\calF'$ with priority-proportional payment functions such that $\veca$ is a clearing state in $\calF$ if and only if $\veca$ is a clearing state in $\calF'$.
\end{lemma}

\begin{proof}
    Consider a bank $v$ in network $\calF$ and the set $B = \bigcup_{e \in E^+(v)} \{x_{e,i} \mid i=1\ldots,k_e\}$ of positive and finite interval borders (without repetition) of payment functions of $v$. Let $k_v = |B|$ and $0 = x_{v,0} < x_{v,1} < \ldots < x_{v,k_v} = L^+(v)$ be the ordered elements of $B$. We define $x_{v,k_v+1} = \infty$. 

    Consider an interval $[x_{v,j-1},x_{v,j})$ for $j = 1,\ldots,k_v$. By construction, for each edge $e \in E^+(v)$ and interval index $j$, there is a unique interval index $j_e \in \{1,\ldots,k_e+1\}$ such that $[x_{v,j-1},x_{v,j}) \subseteq [x_{e,j_e-1}, x_{e,j_e})$. This allows to subdivide the existing interval structure for each edge. We refine every piecewise-linear payment function $p_e$ to have exactly the intervals $[x_{v,j-1},x_{v,j})$ for $j = 1,\ldots,k_v$, as well as the interval $[x_{v,k_v}, x_{v,k_v+1}) = [L^+(v),\infty)$. More precisely,
    \[
        p_e(a) = m_{e,j_e} \cdot (a - p_e(x_{v,j-1})) + p_e(x_{v,j-1})  \qquad \text{ for any } a \in [x_{v,j-1}, x_{v,j}) \subseteq [x_{e,j_e-1}, x_{e,j_e}).
    \]

    We represent each edge $e = (v,w) \in E^+(v)$ in $\calF'$ by ``auxiliary'' edges $e_j$ between $v$ and $w$, for $j = 1,\ldots,k_v$. The set $E_j$ contains exactly the auxiliary edge $e_j$, for each $e \in E^+(v)$. The payment functions are exactly the marginal payments that $v$ would assign to $e$ in the interval $[x_{v,j-1},x_{v,j})$, i.e.,
    \[
        p_{e_j}(a) = 
        \begin{cases}
            0 & \text{if } a < x_{v,j-1} \\
            p_e(a)-p_e(x_{v,j-1}) & \text{if } a \in [x_{v,j-1},x_{v,j}) \\
            p_e(x_{v,j})-p_e(x_{v,j-1}) & \text{if } a \ge x_{v,j} \\
        \end{cases}
    \]
    For every asset value $a_v = a \in [0,L^+(v))$, and $j$ such that $a \in [x_{v,j-1}, x_{v,j})$, we see
    \begin{align*}
        p_e(a) &= (0 - p_e(x_{v,j-1})) + (p_e(a) - p_e(x_{v,j-1})) + 0 \\
        &= \sum_{i < j} (p_e(x_{v,i}) - p_e(x_{v,i-1})) + (p_e(a) - p_e(x_{v,j-1})) + \sum_{i > j} 0 \\
        &= \sum_{i = 1}^{k_v} p_{e_i}(a)
    \end{align*}
    i.e., bank $v$ sends the same payments to $w$ in both $\calF$ and $\calF'$. Clearly for $a_v \ge L^+(v)$ we also have $p_e(a) = L^+(v) = p_{e_{k_v}} - 0 = \sum_{i=1}^{k_v} p_{e_i}(a)$. Therefore, both networks have equivalent payment functions and they must also have the same clearing states.
    
    It remains to define appropriate liabilities $\ell_{e_j}$ such that $\sum_{j} \ell_{e_j} = \ell_e$ and $p_{e_j}$ become priority-proportional. We set 
    \[
        \ell_{e_j} = m_{e,j} \cdot (x_{v,j}-x_{v,j-1}).
    \]
    For $a \ge L^+(v)$ this gives 
    \[
         p_{e_i}(a) = p_e(x_{v,i}) - p_e(x_{v,i-1}) = m_{e,i} \cdot (x_{v,i} - x_{v,i-1}) = \ell_{e_i},
    \]
    for every $i=1,\ldots,k_v$, and, hence, 
    \[
        \ell_e = p_e(a) = \sum_{i=1}^{k_v} p_{e_i}(a) = \sum_{i=1}^{k_v} \ell_{e_i}.
    \] 
    Moreover, for $a \in [0,L^+(v))$
    \begin{align*}
        p_{e_j}(a) &= p_e(a) - p_e(x_{v,j-1}) \\
        &= m_{e,j} \cdot (a - x_{e,j-1}) \\
        &= \frac{m_{e,j}}{\sum_{e' \in E^+(v)} m_{e',j}}  \cdot (a - x_{e,j-1}) && \text{(by~\eqref{eq:rowStochastic}) and $a < L^+(v))$}\\
        &= \frac{\ell_{e_j}}{\sum_{e' \in E^+(v)} \ell_{e'_j}} \cdot (a - x_{e,j-1}) && \text{(since $0 < (x_{v,j} - x_{v,j-1}) < \infty$ for $j \le k_v$)}\\
        &= \frac{\ell_{e_j}}{\sum_{e'_j \in E_j} \ell_{e'_j}} \cdot (a - x_{e,j-1}) 
    \end{align*}
    Therefore, bank $v$ has priority-proportional payments in $\calF'$. 

    Finally, strictly speaking the network $\calF'$ is a \emph{multi-}graph with $k_v$ parallel edges for each $(v,w) \in E$. This can be avoided by further adjusting each edge $e_j$ as follows. Add an auxiliary bank $v_{e,j}$ and split $e_j$ into two edges $e_{j,1} = (v,v_{e,j})$ and $e_{j,2} = (v_{e,j},w)$. $v_{e,j}$ has no external assets and no default cost. $e_{j,1}$ has the same liability and payment function as $e_j$. The liability of $e_{j,2}$ is infinite. Since $e_{j,2}$ is the unique outgoing edge of $v_{e,j}$, the payment function must be $p_{e_{j,2}}(a) = a$. 
    
    It is straightforward to see that $p_e(a) = \sum_{i=1}^{k_v} p_{e_i}(a) = \sum_{i=1}^{k_v} p_{e_{i,1}}(a) = \sum_{i=1}^{k_v} p_{e_{i,2}}(a)$, i.e., given the same assets, $v$ pays exactly the same amount to $w$ as in $\calF$. Hence, the clearing states are equivalent to the ones in $\calF$ (augmented by asset values $a_{v_{e,j}} = p_{e_j}(a)$ for each auxiliary bank). 
\end{proof}

The lemma yields a simple polynomial-time transformation to obtain the network $\calF'$ with $n + O(km)$ banks and $O(km)$ edges. For $\calF'$ the algorithm from~\cite{KanellopoulosKZ24} computes $\check{\veca}$ using an extension of the standard fictitious default algorithm to compute $\check{\veca}$ for proportional clearing~\cite{RogersV13}. Note that the algorithm strictly speaking does not run in finite time since it involves computing payments in each iteration by an iterative procedure that repeatedly solves a system of non-linear equations and converges monotonically from above (see Algorithm 2 line 5). However, it is not difficult to see that the iterative procedure can be implemented in polynomial time by solving a sequence of feasibility LPs and a monotone descent into the priority classes per bank.

\begin{corollary}
    \label{cor:priorityProp}
    There is a polynomial-time algorithm to compute the maximal clearing state in every financial network with piecewise-linear (or priority-proportional) payment functions.
\end{corollary}
\begin{proof}
    The procedure is an efficient implementation of the top-iteration. Suppose each bank $v$ has a priority-proportional payment function. We maintain a counter $r_v \in \{0,1,\ldots,k_v\}$ that indicates the highest edge class which is supposed to be completely paid for by bank $v$. Intuitively, we decide if there is a feasible clearing state for the current counters $\vecr$, i.e., such that for each bank $v$ classes $1,\ldots,r_v$ are completely paid, class $r_v + 1$ is paid proportionally using the remaining assets, and classes $r_v + 2,\ldots, k_v$ are not paid at all. We term this a \emph{clearing state at $\vecr$}. If we realize that this is impossible, we monotonically decrease some of the counters in the next iteration.

    Initially, we set $r_v = k_v$ for all $v \in V$, i.e., we assume all banks are solvent. More generally, consider an arbitrary iteration with a vector $\vecr$. We set $p_e(\veca) = \ell_e$ for all $1 \le i \le r_v$ and each $e \in E_i$. If $r_v = k_v$, then $v$ is meant to be solvent, so any clearing state at $\vecr$ must deliver sufficient assets to $v$, i.e., 
    \[
        a_v = a_v^{(x)} + \sum_{(u,v) \in E^+(v)} p_e(a_u) \ge x_{v,r_v} = L^+(v).
    \]
    If $r_v < k_v$, then for each $r_v+2 \le i \le k_v$ we set $p_e(\veca) = 0$ for each $e \in E_i$. For the remaining class $r_v+1 \le k_v$, there must be non-negative assets available, i.e., 
    \[
        a_v = \alpha_v a_v^{(x)} + \beta_v \sum_{(u,v) \in E^+(v)} p_e(a_u) \ge x_{v,r_v}.
    \]
    To check whether a clearing state at $\vecr$ exists, we relax these constraints by an offset $d_v \ge 0$ for each $v \in V$. With the relaxed constraints we compose the natural feasibility LP~\eqref{eq:FeasibleLP} with the objective to minimize the offsets $\sum_v d_v$. 
    
    If $\check{\veca}$ is a clearing state at $\vecr$, LP \eqref{eq:FeasibleLP} allows a solution without offsets and hence has an optimal value of 0. Otherwise, we maintain by induction the invariant that for any optimal solution $(\veca^*,\vecd^*)$ of LP~\eqref{eq:FeasibleLP} we have $\check{\veca} \le \veca^* + \vecd^*$ coordinate-wise with a strict inequality for at least one entry. Since there is no clearing state at $\vecr$, the available assets of at least one bank $v$ cannot suffice to pay the liabilities of classes $1,\ldots,r_v$ in full. Thus, we require some offset $d_v^* > 0$ to fulfill the corresponding constraint, and the optimal value of LP~\eqref{eq:FeasibleLP} becomes strictly positive.

    \begin{equation}
        \label{eq:FeasibleLP}
        \begin{aligned}
           \text{Min.\ } \; &d_v &\\
           \text{s.t.\ } \; & a_v = a_v^{(x)} + \sum_{(u,v) \in E^+(v)} p_e & \text{for all $v$   with $r_v = k_v$}\\
            &a_v = \alpha_v a_v^{(x)} + \beta_v \sum_{(u,v) \in E^+(v)} p_e & \text{for all $v$ with $r_v < k_v$}\\
            &a_v + d_v \ge x_{v,r_v} & \text{for all $v \in V$}\\    
            &p_e = 0  & \text{for all $e \in E_i$ with $i \ge r_v+2$}\\
            &p_e = \ell_e  & \text{for all $e \in E_i$ with $i \le r_v$}\\
            &p_e = \frac{\ell_e}{\sum_{e' \in E_{r_v+1}} \ell_{e'}} \cdot (a_v + d_v - x_{v,r_v}) & \text{for all $e \in E_{r_v+1}$}\\
            &p_e \ge 0 & \text{for all $e \in E$}\\
            &d_v \ge 0 & \text{for all $v \in V$}
        \end{aligned}
    \end{equation}

    When the optimal value is strictly positive, we must decrease $\vecr$. Consider any bank $v$ with strictly positive offset $d^*_v > 0$ in the optimal solution. Clearly, if $a_v^* + d_v^* > x_{v,r_v}$ we can reduce $\delta_v = a_v^* + d_v^* - x_{v,r_v}$ and raise the offsets of the out-neighbors in $E_{r_v+1}$ by the proportional amount $\delta_v \cdot \ell_e/(\sum_{e' \in E_{r_v+1}} \ell_{e'})$. This maintains feasibility and the objective function value. Iterating this argument, and using optimality of $(\veca^*,\vecd^*)$, we have w.l.o.g.\ that $d_v^* > 0$ implies $a_v^* + d_v^* = x_{v,r_v}$. Since by our invariant $\check{a}_v \le a_v^* + d_v^* = x_{v,r_v}$, it is possible (and potentially necessary) to decrease $r_v$ by (at least) 1. Hence, we decrease $r_v$ by 1 and start the next iteration. Note that this adjustment maintains the invariant.
    
    Thus, after at most $\sum_{v \in V} k_v$ iterations we reach $\vecr$ such that $\check{\veca}$ is a clearing state at $\vecr$. Then LP~\eqref{eq:FeasibleLP} has optimal value 0. We can then compute $\check{\veca}$ by solving an adjusted version of LP~\eqref{eq:FeasibleLP}. We remove all offset variables $d_v$ from the constraints and replace the objective function by maximizing $\sum_{v \in V} a_v$. Overall, the running time is polynomial.
\end{proof}

\section{Claims Trades for Minimal Clearing States}
\label{sec:trade}

In this section, we consider the properties of \emph{claims trading}, a decentralized network adjustment, when applied with \emph{minimal} clearing. In claims trading, a bank $w$ can buy an edge $(u,v)$ by transferring some of its external assets to $v$. Throughout this section, we focus on networks without default cost, i.e., $\alpha_v = \beta_v = 1$ for all $v \in V$.

We start by showing that in a claims trade, it is impossible that \emph{both} creditor $v$ and buyer $w$ \emph{strictly} improve their assets. We prove this result for financial networks in which all payment functions are \emph{strictly} monotone, i.e., $p_e(a_u) < p_e(a_u + \varepsilon)$ for every $\varepsilon \in (0,\ell_e - a_u]$. Our proof uses the following lemma, which states novel properties of uniqueness for clearing states in such networks.
\begin{lemma}
    \label{lem:characterize}
    Consider a bank $v$ in a financial network $\mathcal{F}$ without default cost and with strictly monotone payment functions. If the clearing states are not unique w.r.t.\ $v$ (i.e., $\hat{a}_v \neq \check{a}_v$), then $\hat{a}_v = 0$.
\end{lemma}
\begin{proof}
Let $d_u = \check{a}_u - \hat{a}_u$ be the difference in payments between the greatest and least fixed point. Moreover, let $d_e = p_e(\check{a}_u) - p_e(\hat{a}_u)$ for every edge $e = (u,w) \in E$. Consider any bank $v$ with $d_v > 0$. Then every node $u$ reachable from $v$ in the active graph $G_{\hat{\veca}}$ must have $d_u > 0$. Moreover, $d$ represents a circulation, so $\sum_{e \in E^+(u)} d_e = d_u = \sum_{e \in E^-(u)} d_e$ for every bank $u \in V$. This shows that at the end of Algorithm~\ref{alg:minClear}, $v$ must be part of a non-sink SCC. Since all payment functions are \emph{strictly} monotone, the active graph $G_\vecb$ is monotonically getting sparser during the execution of Algorithm~\ref{alg:minClear} due to solvency of banks. Moreover, since there is no default cost, the set of sinks in the network is monotonically growing. Thus, if $v$ is part of a non-sink SCC in the end, it must be part of a non-sink SCC throughout the entire execution of the algorithm. However, once assets of $v$ get raised, all non-sink SCCs reachable from $v$ must be flooded. As a consequence, $v$ can only be part of such an SCC at the end of the algorithm when $\hat{a}_v = 0$.
\end{proof}
%
%
\begin{theorem}
    Consider a financial network $\mathcal{F}$ without default cost and with strictly monotone payment functions. There exists no claims trade that strictly improves the assets of \emph{both} a buyer~$w$ and a creditor~$v$ w.r.t.\ the minimal clearing state.
\end{theorem}
\begin{proof}   
    Consider a claims trade with claim $(u,v)$ and buyer $w$. In order to pay any return, $w$ must have external assets $a^{(x)}_w > 0$. By Lemma~\ref{lem:characterize}, this implies $\hat{a}_w = \check{a}_w$. 
    
    First, suppose that $\hat{a}_v < \check{a}_v$. Then, Lemma~\ref{lem:characterize} implies $\hat{a}_v = 0$. As discussed in the proof of the lemma, $v$ is part of a non-sink SCC in $G_{\hat{\veca}}$. Moreover, $w$ must be unreachable from $v$ in $G_{\hat{\veca}}$. Since all payment functions are strictly monotonic, the active graph becomes only sparser for larger payments. Hence, $w$ remains unreachable from $v$ even when $v$ has higher assets. Consequently, it is impossible for $w$ to recover any portion of the return $\rho$ paid to $v$ by larger incoming payments. Thus, there cannot be a creditor-positive trade when $\hat{a}_v < \check{a}_v$.

    Second, suppose that $\hat{a}_v = \check{a}_v$. Now, if there was such a trade, w.r.t. $\hat{\veca}$, this trade would also strictly improve the assets of both parties w.r.t.\ $\check{\veca}$. This is impossible~\cite{HoeferVW24}.
\end{proof}

\subsection{Computing Optimal Creditor-Positive Trades}

We focus on computing \emph{creditor-positive} trades in this subsection. The main result is that existence of such trades can be decided in polynomial time for (even non-strictly) monotone, piecewise-linear payment functions. Moreover, an optimal creditor-positive return can be computed in polynomial time. 

Even if there is a creditor-positive trades, it is not directly obvious that an \emph{optimal} such return must exist as well. For example, consider the related problem of \emph{cash injection}. The goal is to allocate $M$ external assets in the network to maximize the total assets. For networks with proportional payments, an optimal solution can be computed in polynomial time when the network is evaluated by the \emph{maximal} clearing state~\cite{KanellopoulosKZ25}. For \emph{minimal} clearing states, it is easy to see that there are simple networks where no optimal solution exists\footnote{Suppose the financial network has two components; a cycle of 2 banks and a path of $n-2$ banks. All edges have liability $1$, all banks have no external assets. Suppose we want to inject $M=1$. Initially, $\hat{\veca} = \vecZero$. Assigning $1-\varepsilon$ to the head of the path and $\epsilon$ to the banks in the cycle yields total assets of $(n-2)(1-\varepsilon) + 2 + \varepsilon$. This expression is maximized for $\varepsilon = 0$, but it applies only when $\varepsilon \in (0,1]$. For $\varepsilon = 0$, the total assets in $\hat{\veca}$ drop to $n-2$. As such, there is no optimal solution.}

We first analyze the structure of the set of creditor-positive returns. Let $\hat{\veca}$ be the pre-trade minimal clearing state and $\rho_{\min} = p_e(\hat{a}_u)$ be the pre-trade payment on claim $e = (u,v)$. 
\begin{lemma}
    \label{lem:tradeStructure}
    The set of creditor-positive returns forms a (possibly empty) interval $(\rho_{\min}, \rho^*]$.
\end{lemma}

\begin{proof}
    We denote by $\hat{\vecb}$ the post-trade minimal clearing state. Suppose there exists any creditor-positive return $\rho$.
    
    We first show that $\rho > \rho_{\min}$. Consider any creditor-positive $\rho$. For convenience, we slightly abuse notation with $p'_{e'} = p_{e'}(\hat{b}_y)$ for the post-trade payments on any edge $e' = (x,y)$ and $p_{e'}$ for the pre-trade payments. We know that $\hat{b}_v > \hat{a}_v$. Also, since any creditor-positive trade represents a Pareto-improvement, the post-trade payments $\vecp' \ge \vecp$ coordinatewise. This implies
    \begin{equation*}
        \hat{b}_v = a_v^{(x)} + \rho + \sum_{e' \in E^-(v) \setminus \{e\}} p'_{e'} \ge a_v^{(x)} + \rho + \sum_{e' \in E^-(v) \setminus \{e\}} p_{e'} = \hat{a}_v + \rho - \rho_{\min},
    \end{equation*}
    so $\rho > \rho_{\min}$.

    For returns $\rho > \rho_{\min}$, consider the post-trade network with a return $\rho_{\min}$ and shift additional external assets from $w$ to $v$. Investing an additional amount of assets directly at $v$ instead of $v$ receiving (parts of) it partly after investment at $w$ can never hurt the assets of $v$ in $\hat{\vecb}$. Consequently, the post-trade assets of $v$ are non-decreasing in $\rho$ over the interval $[\rho_{\min}, \infty)$. By the same argument, the post-trade assets of $w$ are non-increasing in $\rho$ over the interval $[\rho_{\min}, \infty)$. Therefore, all creditor-positive returns must form a consecutive interval.

    For $\rho = \rho_{\min}$ it holds $\hat{\vecb} = \hat{\veca}$, since we simply exchange the same amount of incoming/external assets of $v$ and $w$, respectively. As such, the interval is open on the left. Let us argue that the interval is closed on the right, i.e., it is $(\rho_{\min}, \rho^*]$ for some $\rho^*$ (if non-empty). Thus, if there is any creditor-positive return, an optimal return $\rho^*$ exists. 
    
    Consider an increasing, converging sequence $\lim_{i\to \infty} \rho^{(i)} = \rho^*$. All returns $\rho^{(i)}$ are creditor-positive. Let $\hat{\vecb}^{(i)}$ be the resulting post-trade assets. Note that $\hat{\vecb^{(i)}}$ is coordinate-wise non-decreasing -- the assets of $w$ remain at $\hat{a}_w$, the assets of $v$ are non-decreasing. As such, the assets in the entire network are non-decreasing. Upon shifting an additional amount of $\varepsilon^{(i)} = \rho^{(i)} - \rho^{(i-1)}$ of external assets to $v$, we decrease the external assets of $w$ by this amount. Since both returns are creditor-positive, $w$ must receive additional incoming payments of $\varepsilon^{(i)}$ in the network. The total assets shall remain at $\hat{a}_w$, which keeps the outgoing assets fixed at $\min\{\hat{a}_w, L^+(w)\}$. Hence, $w$ cannot be part of additional flooded components, so the additional incoming payments of $\varepsilon^{(i)}$ must originate from the $\varepsilon^{(i)}$ external assets invested at $v$. Since all payment functions are continuous and non-decreasing, there must be a unique largest value $\rho^*$ for which the additional external assets invested at $v$ arrive completely at $w$. This is the limit $\rho^*$, and it is also a creditor-positive return. 
\end{proof}
Based on this structural insight, we proceed to show the main result of this section.
\begin{theorem}
Consider a financial network $\calF$ without default cost and with piecewise-linear payment functions. For a given claim $e = (u,v)$ and a buyer $w$, it can be decided in polynomial-time if a creditor-positive trade exists. If the trade exists, the optimal creditor-positive return can be computed in polynomial time. 
\end{theorem}

\begin{proof}
    Suppose we use a return of $\rho_{\min}$. As observed above, the resulting minimal clearing state is $\hat{\vecb} = \hat{\veca}$. Consider the active graph $G_{\hat{\vecb}}$. Let us increase the return $\rho$ by a sufficiently small value $\delta$ and denote the resulting minimal clearing state by $\hat{\vecb}^{(\delta)}$. 
    
    First, suppose we do not pass an interval border on any edge, i.e., $G_{\hat{\vecb}^{(\delta)}}' = G_{\hat{\vecb}}$. Then the effect on the minimal clearing state must be linear, i.e., the change in the assets is given by
    \[
        \vecd =  \hat{\vecb}^{(\delta)} - \hat{\vecb} = \delta \cdot \vece_v - \delta \cdot \vece_w + \vecd \; \matM,
    \]
    where $\vece_v$ and $\vece_w$ are unit vectors with an entry of 1 for $v$ and $w$, respectively, and $\matM$ is the matrix of all slopes of edges in $G_{\hat{\vecb}}$ (c.f.\ \eqref{eq:linearIncrease} above). $\vecd$ is linear in $\delta$ with slopes
    \[
        \vecs = (\matI - \matM)^{-1} \cdot (\vece_v - \vece_w).
    \]
    As observed in the proof of Lemma~\ref{lem:tradeStructure}, we have $s_w \le 0 \le s_v$. If $s_w = 0$, we can increase the return and it will stay creditor-positive. Given that the assets of $w$ remain $a_w$, this Pareto-improve all assets in the network (and, thus, $\vecs \ge \vecZero$ coordinate-wise).
    
    An increase can be implemented very similarly as in Algorithm~\ref{alg:minClear}. We adjust all payments linearly until we reach an interval border for some payment function. 

    Second, suppose $G_{\hat{\vecb}}$ passes an interval border upon increase of $\rho$. More precisely, there is a positive slope of $v$ and slope $0$ for $w$, and raising the return requires a change in the active graph. If any larger creditor-positive return exists, it would further raise the assets of $v$ (and keep the assets of $w$ at $a_w$). This means that we first have to apply the flooding operation on all non-sink SCCs reachable from $v$ as in Algorithm~\ref{alg:minClear} above. Then, we again determine the slopes of further increase as before and check if the slope of $w$ remains 0 or becomes negative. 

    We check the existence of a creditor-positive return as follows: Compute the slopes $\vecs$ for $G_{\hat{\vecb}}$. The initial slopes must be $s_v > 0 = s_w$, otherwise no creditor-positive return exists. If they are, and $G_{\hat{\vecb}}$ is located at an interval border, flood the appropriate SCCs and check the slopes again. If they continue to be $s_v > 0 = s_w$, a creditor-positive return exists; otherwise not.

    If there is a creditor-positive return, we can search iteratively by increasing the return, changing the active graph, and flooding components as long as the resulting slopes are $s_v > 0 = s_w$. It requires repeatedly solving systems of linear equations. The approach is very similar to Algorithm~\ref{alg:minClear}, and we obtain the same asymptotic upper bound on the running time. Alternatively, we can binary search over the interval $[\rho_{\min},\min\{a_v^{(x)}, \ell_e\}]$. Once we find a return that is creditor-positive, we refine it by computing the slopes and increasing the return until the active graph hits the next interval border. This approach is faster if $\rho^*$ is large and there are many interval borders for small creditor-positive returns. It can be slower if there are few interval borders and $\rho^* \ll \min\{a_v^{(x)}, \ell_e\}$.
\end{proof}

\bibliographystyle{abbrv} 
\bibliography{financialBib}

\clearpage

\appendix
\section{Examples}
\label{app:examples}

\begin{example}
    \label{exmpl:bottom iteration} \rm
    The bottom iteration has a rather intuitive meaning.
    It corresponds to a process in which banks begin to pay off their debt with their external assets and, in each step, use those new assets that they just received.
    Consider the following network\\

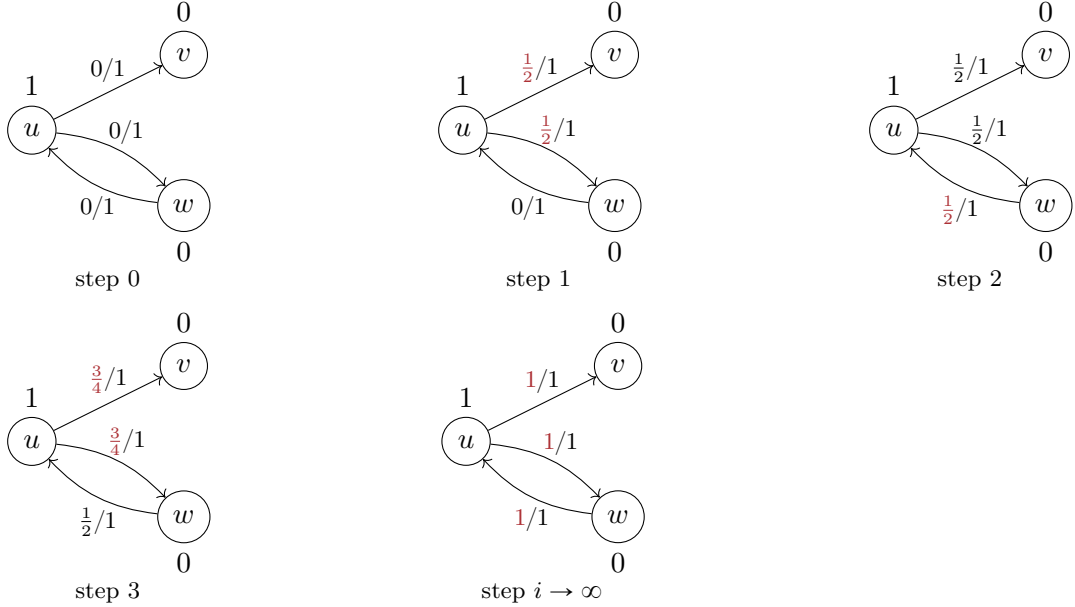
    \begin{minipage}{0.34\textwidth}
    \begin{tikzpicture}
        \node[draw, circle, label={above:\(1\)}] (u) at (0,0) {\(u\)};
        \node[draw, circle, label={above:\(0\)}] (v) at (2,1) {\(v\)};
        \node[draw, circle, label={below:\(0\)}] (w) at (2,-1) {\(w\)};

        \path[->]
            (u) edge[bend left=20] node[pos=0.6,above] {\footnotesize\(0/1\)} (w)
            (w) edge[bend left=20] node[pos=0.5,below] {\footnotesize\(0/1\)} (u)
            (u) edge node[pos=0.5,above] {\footnotesize\(0/1\)} (v);

        \node (label) at (1,-2) {\footnotesize step 0};
    \end{tikzpicture}
    \end{minipage}
    \begin{minipage}{0.34\textwidth}
    \begin{tikzpicture}
        \node[draw, circle, label={above:\(1\)}] (u) at (0,0) {\(u\)};
        \node[draw, circle, label={above:\(0\)}] (v) at (2,1) {\(v\)};
        \node[draw, circle, label={below:\(0\)}] (w) at (2,-1) {\(w\)};

        \path[->]
            (u) edge[bend left=20] node[pos=0.6,above] {\footnotesize\({\color{Maroon}\frac{1}{2}}/1\)} (w)
            (w) edge[bend left=20] node[pos=0.5,below] {\footnotesize\(0/1\)} (u)
            (u) edge node[pos=0.5,above] {\footnotesize\({\color{Maroon}\frac{1}{2}}/1\)} (v);

        \node (label) at (1,-2) {\footnotesize step 1};
    \end{tikzpicture}
    \end{minipage}
    \begin{minipage}{0.34\textwidth}
    \begin{tikzpicture}
        \node[draw, circle, label={above:\(1\)}] (u) at (0,0) {\(u\)};
        \node[draw, circle, label={above:\(0\)}] (v) at (2,1) {\(v\)};
        \node[draw, circle, label={below:\(0\)}] (w) at (2,-1) {\(w\)};

        \path[->]
            (u) edge[bend left=20] node[pos=0.6,above] {\footnotesize\(\frac{1}{2}/1\)} (w)
            (w) edge[bend left=20] node[pos=0.5,below] {\footnotesize\({\color{Maroon}\frac{1}{2}}/1\)} (u)
            (u) edge node[pos=0.5,above] {\footnotesize\(\frac{1}{2}/1\)} (v);

        \node (label) at (1,-2) {\footnotesize step 2};
    \end{tikzpicture}
    \end{minipage}

    \begin{minipage}{0.34\textwidth}
    \begin{tikzpicture}
        \node[draw, circle, label={above:\(1\)}] (u) at (0,0) {\(u\)};
        \node[draw, circle, label={above:\(0\)}] (v) at (2,1) {\(v\)};
        \node[draw, circle, label={below:\(0\)}] (w) at (2,-1) {\(w\)};

        \path[->]
            (u) edge[bend left=20] node[pos=0.6,above] {\footnotesize\({\color{Maroon}\frac{3}{4}}/1\)} (w)
            (w) edge[bend left=20] node[pos=0.5,below] {\footnotesize\(\frac{1}{2}/1\)} (u)
            (u) edge node[pos=0.5,above] {\footnotesize\({\color{Maroon}\frac{3}{4}}/1\)} (v);

        \node (label) at (1,-2) {\footnotesize step 3};
    \end{tikzpicture}
    \end{minipage}
    \begin{minipage}{0.34\textwidth}
    \begin{tikzpicture}
        \node[draw, circle, label={above:\(1\)}] (u) at (0,0) {\(u\)};
        \node[draw, circle, label={above:\(0\)}] (v) at (2,1) {\(v\)};
        \node[draw, circle, label={below:\(0\)}] (w) at (2,-1) {\(w\)};

        \path[->]
            (u) edge[bend left=20] node[pos=0.6,above] {\footnotesize\({\color{Maroon}1}/1\)} (w)
            (w) edge[bend left=20] node[pos=0.5,below] {\footnotesize\({\color{Maroon}1}/1\)} (u)
            (u) edge node[pos=0.5,above] {\footnotesize\({\color{Maroon}1}/1\)} (v);

        \node (label) at (1,-2) {\footnotesize step $i \to \infty$};
    \end{tikzpicture}
    \end{minipage}
    %
    %
    \vspace{.5cm}\\
    In this case each edge will see a payment of
    $\sum_{i = 0}^n \left( \frac{1}{2} \right)^i$ in step \(2n\).
    This is a geometric series, so all edges will have payments of \(1\) in the minimal clearing state.
\end{example}

\begin{example} 
    \label{exmpl:bottom iteration 2} \rm
    We consider the bottom iteration in an example with default rates, where in the limit we do not reach the minimal clearing state. Consider the following network, and let $\alpha_v = \alpha_w = \beta_v = \beta_w = 1/2$.
    
    \bigskip
    \begin{minipage}{0.34\textwidth}
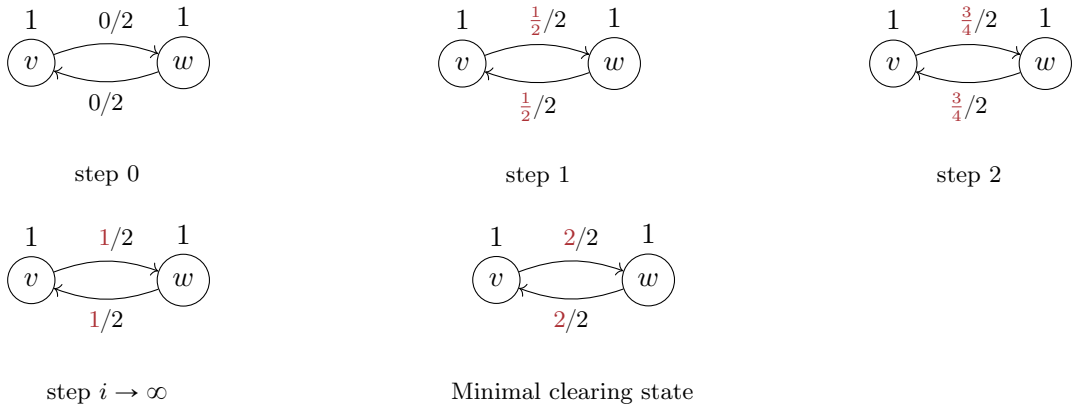
    \begin{tikzpicture}
        \node[draw, circle, label={above:\(1\)}] (v) at (0,0) {\(v\)};
        \node[draw, circle, label={above:\(1\)}] (w) at (2,0) {\(w\)};

        \path[->]
            (v) edge[bend left=20] node[pos=0.6,above] {\footnotesize\(0/2\)} (w)
            (w) edge[bend left=20] node[pos=0.5,below] {\footnotesize\(0/2\)} (v);

        \node (label) at (1,-1.5) {\footnotesize step 0};
    \end{tikzpicture}
    \end{minipage}
    \begin{minipage}{0.34\textwidth}
    \begin{tikzpicture}
        \node[draw, circle, label={above:\(1\)}] (v) at (0,0) {\(v\)};
        \node[draw, circle, label={above:\(1\)}] (w) at (2,0) {\(w\)};

        \path[->]
            (v) edge[bend left=20] node[pos=0.6,above] {\footnotesize\({\color{Maroon}\frac{1}{2}}/2\)} (w)
            (w) edge[bend left=20] node[pos=0.5,below] {\footnotesize\({\color{Maroon}\frac{1}{2}}/2\)} (v);

        \node (label) at (1,-1.5) {\footnotesize step 1};
    \end{tikzpicture}
    \end{minipage}
    \begin{minipage}{0.34\textwidth}
    \begin{tikzpicture}
        \node[draw, circle, label={above:\(1\)}] (v) at (0,0) {\(v\)};
        \node[draw, circle, label={above:\(1\)}] (w) at (2,0) {\(w\)};

        \path[->]
            (v) edge[bend left=20] node[pos=0.6,above] {\footnotesize\({\color{Maroon}\frac{3}{4}}/2\)} (w)
            (w) edge[bend left=20] node[pos=0.5,below] {\footnotesize\({\color{Maroon}\frac{3}{4}}/2\)} (v);

        \node (label) at (1,-1.5) {\footnotesize step 2};
    \end{tikzpicture}
    \end{minipage}

    \medskip
    \begin{minipage}{0.34\textwidth}
    \begin{tikzpicture}
        \node[draw, circle, label={above:\(1\)}] (v) at (0,0) {\(v\)};
        \node[draw, circle, label={above:\(1\)}] (w) at (2,0) {\(w\)};

        \path[->]
            (v) edge[bend left=20] node[pos=0.6,above] {\footnotesize\({\color{Maroon} 1}/2\)} (w)
            (w) edge[bend left=20] node[pos=0.5,below] {\footnotesize\({\color{Maroon} 1}/2\)} (v);

        \node (label) at (1,-1.5) {\footnotesize step $i \to \infty$};
    \end{tikzpicture}
    \end{minipage}
    \begin{minipage}{0.34\textwidth}
    \begin{tikzpicture}
        \node[draw, circle, label={above:\(1\)}] (v) at (0,0) {\(v\)};
        \node[draw, circle, label={above:\(1\)}] (w) at (2,0) {\(w\)};

        \path[->]
            (v) edge[bend left=20] node[pos=0.6,above] {\footnotesize\({\color{Maroon} 2}/2\)} (w)
            (w) edge[bend left=20] node[pos=0.5,below] {\footnotesize\({\color{Maroon} 2}/2\)} (v);

        \node (label) at (1,-1.5) {\footnotesize Minimal clearing state};
    \end{tikzpicture}
    \end{minipage}
    \vspace{.5cm}\\
    Throughout the iteration, both $v$ and $w$ remain insolvent, the payments on each of their incoming edges remain strictly below 1. In the limit, these payments become 1, and exactly at this point both banks become solvent. Consequently, in the minimal clearing state they are both solvent and clear all debt with payments of 2 per edge. Consequently, the limit of the iteration is not the minimal clearing state.
\end{example}

\begin{example}
    \label{exmpl:runthrough} \rm
    We will apply our algorithm to compute a minimal clearing state in the following network with edge-ranking payment functions. For bank \(v\) the edge to \(w\) has higher priority than the one to \(y\). Each bank $v$ is labeled with $b^{(x)}_v / a^{(x)}_v$ and each edge $e = (u,v)$ with $p_e(b_u) / \ell_e$.

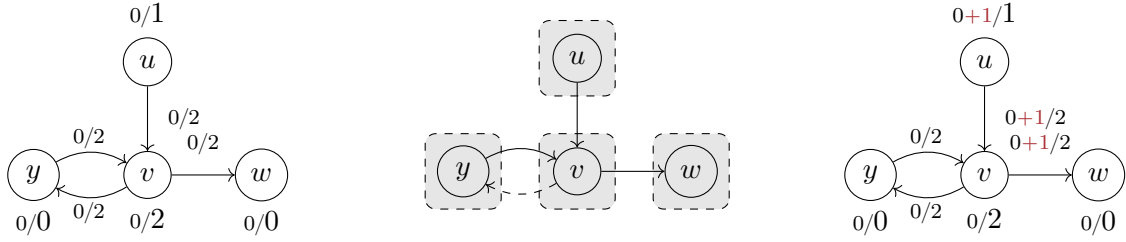
\begin{minipage}{0.33\textwidth}
\begin{tikzpicture}\
    \node[draw, circle, label={above:\({\scriptstyle 0/}1\)}] (u) at (0,1.5) {\(u\)};
    \node[draw, circle, label={below:\({\scriptstyle 0/}2\)}] (v) at (0,0) {\(v\)};
    \node[draw, circle, label={below:\({\scriptstyle 0/}0\)}] (w) at (1.5,0) {\(w\)};
    \node[draw, circle, label={below:\({\scriptstyle 0/}0\)}] (y) at (-1.5,0) {\(y\)};

    \draw[->] (u) -- (v) node[midway, label={right:\(\scriptstyle0/2\)}] {};
    \draw[->] (v) -- (w) node[midway, label={above:\(\scriptstyle0/2\)}] {};
    \draw[->] (v) edge[bend left] (y) node[label={right:\(\scriptstyle0/2\)}] at (-1.25,0.5) {};
    \draw[->] (y) edge[bend left] (v) node[label={right:\(\scriptstyle0/2\)}] at (-1.25,-0.5) {};
\end{tikzpicture}
\end{minipage}
\begin{minipage}{0.33\textwidth}
\begin{tikzpicture}
    \draw[rounded corners,fill=gray!20, dashed] (-0.5,-0.5)--(-0.5,0.5)--(0.5,0.5)--(0.5,-0.5)--cycle;
    \draw[rounded corners,fill=gray!20, dashed] (-2,-0.5)--(-2,0.5)--(-1,0.5)--(-1,-0.5)--cycle;
    \draw[rounded corners,fill=gray!20, dashed] (-0.5,1)--(-0.5,2)--(0.5,2)--(0.5,1)--cycle;
    \draw[rounded corners,fill=gray!20, dashed] (2,-0.5)--(2,0.5)--(1,0.5)--(1,-0.5)--cycle;
    
    \node[draw, circle] (u) at (0,1.5) {\(u\)};
    \node[draw, circle] (v) at (0,0) {\(v\)};
    \node[draw, circle] (w) at (1.5,0) {\(w\)};
    \node[draw, circle] (y) at (-1.5,0) {\(y\)};

    \draw[->] (u) -- (v) node {};
    \draw[->] (v) -- (w) node {};
    \draw[->, dashed] (v) edge[bend left] (y) node at (-1.25,0.5) {};
    \draw[->] (y) edge[bend left] (v) node at (-1.25,-0.5) {};
\end{tikzpicture}
\end{minipage}
\begin{minipage}{0.33\textwidth}
\begin{tikzpicture}
    \node[draw, circle, label={above:\({\scriptstyle 0{\color{Maroon} + 1}/}1\)}] (u) at (0,1.5) {\(u\)};
    \node[draw, circle, label={below:\({\scriptstyle 0/}2\)}] (v) at (0,0) {\(v\)};
    \node[draw, circle, label={below:\({\scriptstyle 0/}0\)}] (w) at (1.5,0) {\(w\)};
    \node[draw, circle, label={below:\({\scriptstyle 0/}0\)}] (y) at (-1.5,0) {\(y\)};

    \draw[->] (u) -- (v) node[midway, label={right:\(\scriptstyle0{\color{Maroon} + 1}/2\)}] {};
    \draw[->] (v) -- (w) node[midway, label={above:\(\scriptstyle0{\color{Maroon} + 1}/2\)}] {};
    \draw[->] (v) edge[bend left] (y) node[label={right:\(\scriptstyle0/2\)}] at (-1.25,0.5) {};
    \draw[->] (y) edge[bend left] (v) node[label={right:\(\scriptstyle0/2\)}] at (-1.25,-0.5) {};
\end{tikzpicture}
\end{minipage}

    \vspace{.25cm}
    \noindent We choose \(u\) as the first bank to insert assets.
    The middle figure shows the strongly connected components.
    The right figure shows the increase of payments and external assets with $\delta^* = 1$.

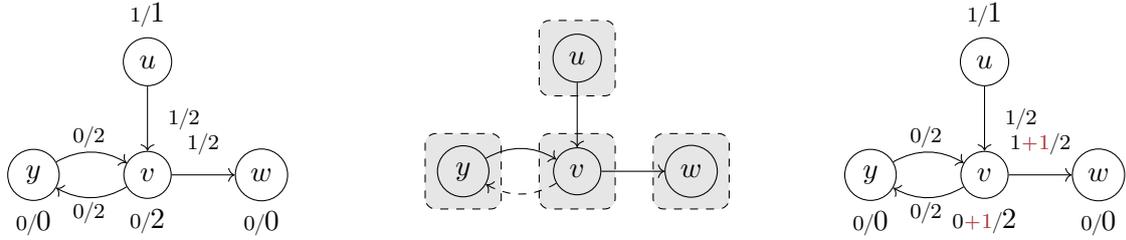
\begin{minipage}{0.33\textwidth}
\begin{tikzpicture}
    \node[draw, circle, label={above:\({\scriptstyle 1/}1\)}] (u) at (0,1.5) {\(u\)};
    \node[draw, circle, label={below:\({\scriptstyle 0/}2\)}] (v) at (0,0) {\(v\)};
    \node[draw, circle, label={below:\({\scriptstyle 0/}0\)}] (w) at (1.5,0) {\(w\)};
    \node[draw, circle, label={below:\({\scriptstyle 0/}0\)}] (y) at (-1.5,0) {\(y\)};

    \draw[->] (u) -- (v) node[midway, label={right:\(\scriptstyle1/2\)}] {};
    \draw[->] (v) -- (w) node[midway, label={above:\(\scriptstyle1/2\)}] {};
    \draw[->] (v) edge[bend left] (y) node[label={right:\(\scriptstyle0/2\)}] at (-1.25,0.5) {};
    \draw[->] (y) edge[bend left] (v) node[label={right:\(\scriptstyle0/2\)}] at (-1.25,-0.5) {};
\end{tikzpicture}
\end{minipage}
\begin{minipage}{0.33\textwidth}
\begin{tikzpicture}
    \draw[rounded corners,fill=gray!20, dashed] (-0.5,-0.5)--(-0.5,0.5)--(0.5,0.5)--(0.5,-0.5)--cycle;
    \draw[rounded corners,fill=gray!20, dashed] (-2,-0.5)--(-2,0.5)--(-1,0.5)--(-1,-0.5)--cycle;
    \draw[rounded corners,fill=gray!20, dashed] (-0.5,1)--(-0.5,2)--(0.5,2)--(0.5,1)--cycle;
    \draw[rounded corners,fill=gray!20, dashed] (2,-0.5)--(2,0.5)--(1,0.5)--(1,-0.5)--cycle;
    
    \node[draw, circle] (u) at (0,1.5) {\(u\)};
    \node[draw, circle] (v) at (0,0) {\(v\)};
    \node[draw, circle] (w) at (1.5,0) {\(w\)};
    \node[draw, circle] (y) at (-1.5,0) {\(y\)};

    \draw[->] (u) -- (v) node {};
    \draw[->] (v) -- (w) node {};
    \draw[->, dashed] (v) edge[bend left] (y) node at (-1.25,0.5) {};
    \draw[->] (y) edge[bend left] (v) node at (-1.25,-0.5) {};
\end{tikzpicture}
\end{minipage}
\begin{minipage}{0.33\textwidth}
\begin{tikzpicture}
    \node[draw, circle, label={above:\({\scriptstyle 1/}1\)}] (u) at (0,1.5) {\(u\)};
    \node[draw, circle, label={below:\({\scriptstyle 0{\color{Maroon} + 1}/}2\)}] (v) at (0,0) {\(v\)};
    \node[draw, circle, label={below:\({\scriptstyle 0/}0\)}] (w) at (1.5,0) {\(w\)};
    \node[draw, circle, label={below:\({\scriptstyle 0/}0\)}] (y) at (-1.5,0) {\(y\)};

    \draw[->] (u) -- (v) node[midway, label={right:\(\scriptstyle1/2\)}] {};
    \draw[->] (v) -- (w) node[midway, label={above:\(\scriptstyle1{\color{Maroon} + 1}/2\)}] {};
    \draw[->] (v) edge[bend left] (y) node[label={right:\(\scriptstyle0/2\)}] at (-1.25,0.5) {};
    \draw[->] (y) edge[bend left] (v) node[label={right:\(\scriptstyle0/2\)}] at (-1.25,-0.5) {};
\end{tikzpicture}
\end{minipage}

    \vspace{.25cm}
    \noindent The next and last bank to have its assets inserted will be \(v\).
    The strongly connected components have not changed.
    This time we can see that inserting all of \(v\)'s external assets at once would cross the breakpoint of \(v\)'s edge to \(w\).
    Hence we have \(\delta^* = 1\).

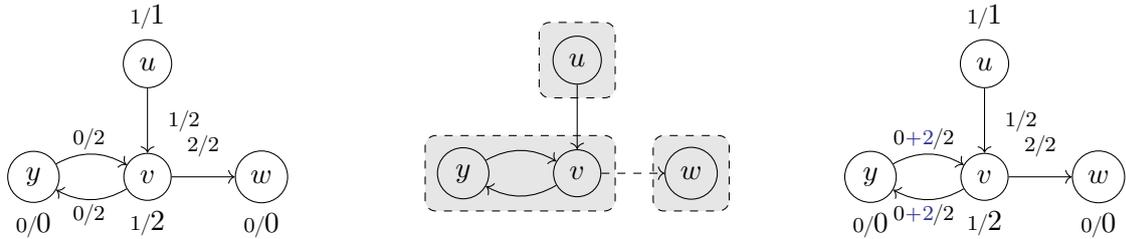
\begin{minipage}{0.33\textwidth}
\begin{tikzpicture}
    \node[draw, circle, label={above:\({\scriptstyle 1/}1\)}] (u) at (0,1.5) {\(u\)};
    \node[draw, circle, label={below:\({\scriptstyle 1/}2\)}] (v) at (0,0) {\(v\)};
    \node[draw, circle, label={below:\({\scriptstyle 0/}0\)}] (w) at (1.5,0) {\(w\)};
    \node[draw, circle, label={below:\({\scriptstyle 0/}0\)}] (y) at (-1.5,0) {\(y\)};

    \draw[->] (u) -- (v) node[midway, label={right:\(\scriptstyle1/2\)}] {};
    \draw[->] (v) -- (w) node[midway, label={above:\(\scriptstyle2/2\)}] {};
    \draw[->] (v) edge[bend left] (y) node[label={right:\(\scriptstyle0/2\)}] at (-1.25,0.5) {};
    \draw[->] (y) edge[bend left] (v) node[label={right:\(\scriptstyle0/2\)}] at (-1.25,-0.5) {};
\end{tikzpicture}
\end{minipage}
\begin{minipage}{0.33\textwidth}
\begin{tikzpicture}
    \draw[rounded corners,fill=gray!20, dashed] (-2,-0.5)--(-2,0.5)--(0.5,0.5)--(0.5,-0.5)--cycle;
    \draw[rounded corners,fill=gray!20, dashed] (-0.5,1)--(-0.5,2)--(0.5,2)--(0.5,1)--cycle;
    \draw[rounded corners,fill=gray!20, dashed] (2,-0.5)--(2,0.5)--(1,0.5)--(1,-0.5)--cycle;
    
    \node[draw, circle] (u) at (0,1.5) {\(u\)};
    \node[draw, circle] (v) at (0,0) {\(v\)};
    \node[draw, circle] (w) at (1.5,0) {\(w\)};
    \node[draw, circle] (y) at (-1.5,0) {\(y\)};

    \draw[->] (u) -- (v) node {};
    \draw[->, dashed] (v) -- (w) node {};
    \draw[->] (v) edge[bend left] (y) node at (-1.25,0.5) {};
    \draw[->] (y) edge[bend left] (v) node at (-1.25,-0.5) {};
\end{tikzpicture}
\end{minipage}
\begin{minipage}{0.33\textwidth}
\begin{tikzpicture}
    \node[draw, circle, label={above:\({\scriptstyle 1/}1\)}] (u) at (0,1.5) {\(u\)};
    \node[draw, circle, label={below:\({\scriptstyle 1/}2\)}] (v) at (0,0) {\(v\)};
    \node[draw, circle, label={below:\({\scriptstyle 0/}0\)}] (w) at (1.5,0) {\(w\)};
    \node[draw, circle, label={below:\({\scriptstyle 0/}0\)}] (y) at (-1.5,0) {\(y\)};

    \draw[->] (u) -- (v) node[midway, label={right:\(\scriptstyle1/2\)}] {};
    \draw[->] (v) -- (w) node[midway, label={above:\(\scriptstyle2/2\)}] {};
    \draw[->] (v) edge[bend left] (y) node[label={right:\(\scriptstyle0{\color{Blue} + 2}/2\)}] at (-1.47,0.5) {};
    \draw[->] (y) edge[bend left] (v) node[label={right:\(\scriptstyle0{\color{Blue} + 2}/2\)}] at (-1.47,-0.5) {};
\end{tikzpicture}
\end{minipage}

    \vspace{.25cm}
    \noindent We can see that the active edges and hence the strongly connected components have changed.
    There is now a flooded region $\{v,y\}$.
    We solve the corresponding LP and increase the payments.
    Afterwards we need to update the strongly connected components.

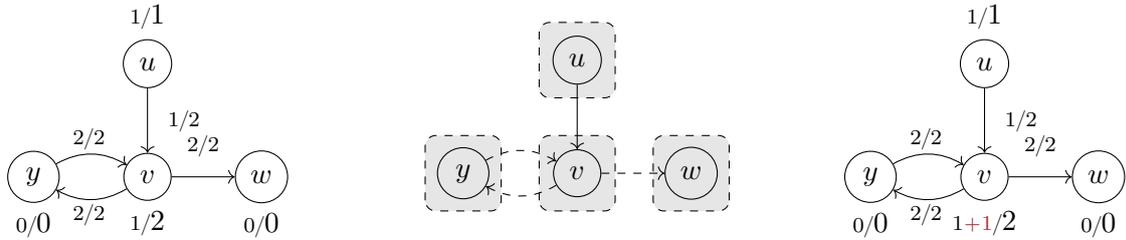
\begin{minipage}{0.33\textwidth}
\begin{tikzpicture}
    \node[draw, circle, label={above:\({\scriptstyle 1/}1\)}] (u) at (0,1.5) {\(u\)};
    \node[draw, circle, label={below:\({\scriptstyle 1/}2\)}] (v) at (0,0) {\(v\)};
    \node[draw, circle, label={below:\({\scriptstyle 0/}0\)}] (w) at (1.5,0) {\(w\)};
    \node[draw, circle, label={below:\({\scriptstyle 0/}0\)}] (y) at (-1.5,0) {\(y\)};

    \draw[->] (u) -- (v) node[midway, label={right:\(\scriptstyle1/2\)}] {};
    \draw[->] (v) -- (w) node[midway, label={above:\(\scriptstyle2/2\)}] {};
    \draw[->] (v) edge[bend left] (y) node[label={right:\(\scriptstyle2/2\)}] at (-1.25,0.5) {};
    \draw[->] (y) edge[bend left] (v) node[label={right:\(\scriptstyle2/2\)}] at (-1.25,-0.5) {};
\end{tikzpicture}
\end{minipage}
\begin{minipage}{0.33\textwidth}
\begin{tikzpicture}
    \draw[rounded corners,fill=gray!20, dashed] (-0.5,-0.5)--(-0.5,0.5)--(0.5,0.5)--(0.5,-0.5)--cycle;
    \draw[rounded corners,fill=gray!20, dashed] (-2,-0.5)--(-2,0.5)--(-1,0.5)--(-1,-0.5)--cycle;
    \draw[rounded corners,fill=gray!20, dashed] (-0.5,1)--(-0.5,2)--(0.5,2)--(0.5,1)--cycle;
    \draw[rounded corners,fill=gray!20, dashed] (2,-0.5)--(2,0.5)--(1,0.5)--(1,-0.5)--cycle;
    
    \node[draw, circle] (u) at (0,1.5) {\(u\)};
    \node[draw, circle] (v) at (0,0) {\(v\)};
    \node[draw, circle] (w) at (1.5,0) {\(w\)};
    \node[draw, circle] (y) at (-1.5,0) {\(y\)};

    \draw[->] (u) -- (v) node {};
    \draw[->, dashed] (v) -- (w) node {};
    \draw[->, dashed] (v) edge[bend left] (y) node at (-1.25,0.5) {};
    \draw[->, dashed] (y) edge[bend left] (v) node at (-1.25,-0.5) {};
\end{tikzpicture}
\end{minipage}
\begin{minipage}{0.33\textwidth}
\begin{tikzpicture}
    \node[draw, circle, label={above:\({\scriptstyle 1/}1\)}] (u) at (0,1.5) {\(u\)};
    \node[draw, circle, label={below:\({\scriptstyle 1{\color{Maroon} + 1}/}2\)}] (v) at (0,0) {\(v\)};
    \node[draw, circle, label={below:\({\scriptstyle 0/}0\)}] (w) at (1.5,0) {\(w\)};
    \node[draw, circle, label={below:\({\scriptstyle 0/}0\)}] (y) at (-1.5,0) {\(y\)};

    \draw[->] (u) -- (v) node[midway, label={right:\(\scriptstyle1/2\)}] {};
    \draw[->] (v) -- (w) node[midway, label={above:\(\scriptstyle2/2\)}] {};
    \draw[->] (v) edge[bend left] (y) node[label={right:\(\scriptstyle2/2\)}] at (-1.25,0.5) {};
    \draw[->] (y) edge[bend left] (v) node[label={right:\(\scriptstyle2/2\)}] at (-1.25,-0.5) {};
\end{tikzpicture}
\end{minipage}

    \vspace{.25cm}
    \noindent Since there are no outgoing active edges from $v$ any more, we get $\vecd = \mathbf{0}$ and $\delta^* = 1$.

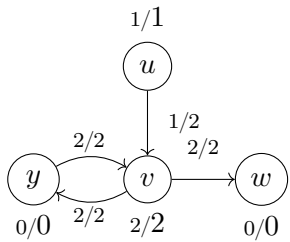
\begin{minipage}{0.33\textwidth}
\begin{tikzpicture}
    \node[draw, circle, label={above:\({\scriptstyle 1/}1\)}] (u) at (0,1.5) {\(u\)};
    \node[draw, circle, label={below:\({\scriptstyle 2/}2\)}] (v) at (0,0) {\(v\)};
    \node[draw, circle, label={below:\({\scriptstyle 0/}0\)}] (w) at (1.5,0) {\(w\)};
    \node[draw, circle, label={below:\({\scriptstyle 0/}0\)}] (y) at (-1.5,0) {\(y\)};

    \draw[->] (u) -- (v) node[midway, label={right:\(\scriptstyle1/2\)}] {};
    \draw[->] (v) -- (w) node[midway, label={above:\(\scriptstyle2/2\)}] {};
    \draw[->] (v) edge[bend left] (y) node[label={right:\(\scriptstyle2/2\)}] at (-1.25,0.5) {};
    \draw[->] (y) edge[bend left] (v) node[label={right:\(\scriptstyle2/2\)}] at (-1.25,-0.5) {};
\end{tikzpicture}
\end{minipage}

    \vspace{.25cm}
    \noindent At this point all external assets have been inserted, and we have successfully computed the minimal clearing state.
\end{example}
\end{document}